\documentclass[a4paper,USenglish]{lipics-v2016}
\usepackage{graphicx}
\usepackage{paralist}
\usepackage{microtype}
\usepackage{amsmath}
\usepackage{amsthm}
\usepackage{amsxtra}
\usepackage{amssymb}
\usepackage{color}

\usepackage{multirow}
\usepackage{pdflscape}

\usepackage{verbatim} 
\usepackage{hyperref}

\newcommand{\de}{d_{G \setminus e}}
\newcommand{\dG}{d_{G}}
\newcommand{\poly}{\operatorname{poly}}

\newcommand{\tO}{\tilde{O}}

\theoremstyle{plain}

\newtheorem{conjecture}[theorem]{Conjecture}

\title{Conditional Hardness for Sensitivity Problems}
\titlerunning{Conditional Hardness for Sensitivity Problems} 

\author[1]{Monika Henzinger \thanks{The research
		leading to these results has received funding from the European
		Research Council under the European Union’s Seventh Framework
		Programme (FP/2007-2013) / ERC Grant Agreement no. 340506.}}
\author[2]{Andrea Lincoln \thanks{Supported by a Stanford Graduate Fellowship.}}
\author[3]{Stefan Neumann \thanks{Supported by the Doctoral Programme ``Vienna Graduate School on Computational Optimization'' which is funded by Austrian Science Fund (FWF, project no.\ W1260-N35). }}
\author[4]{Virginia Vassilevska Williams \thanks{VVW and AL were supported by NSF Grants CCF-1417238, CCF-1528078 and CCF-1514339, and BSF Grant BSF:2012338.}}
\affil[1]{University of Vienna, Faculty of Computer Science, Vienna, Austria\\
  \texttt{monika.henzinger@univie.ac.at}}
\affil[2]{Computer Science Department, Stanford University, Stanford, USA\\
  \texttt{andreali@cs.stanford.edu}}
\affil[3]{University of Vienna, Faculty of Computer Science, Vienna, Austria\\
  \texttt{stefan.neumann@univie.ac.at}}
\affil[4]{Computer Science Department, Stanford University, Stanford, USA\\
  \texttt{virgi@cs.stanford.edu}}
\authorrunning{M.\ Henzinger, A.\ Lincoln, S.\ Neumann and V.\ Vassilevska Williams} 

\Copyright{Monika Henzinger, Andrea Lincoln, Stefan Neumann and Virginia Vassilevska Williams}


\begin{document}

\maketitle
\begin{abstract}
	In recent years it has become popular to study dynamic problems in a sensitivity
	setting: Instead of allowing for an arbitrary sequence of updates, the sensitivity
	model only allows to apply batch updates of small size to the \emph{original} input data.
	The sensitivity model is particularly appealing since recent
	strong conditional lower bounds ruled out fast algorithms for
	many dynamic problems, such as shortest paths, reachability, or subgraph connectivity.

	In this paper we prove conditional lower bounds for these and additional
	problems in a sensitivity setting. For example, we show that
	under the Boolean Matrix Multiplication (BMM) conjecture combinatorial algorithms
	cannot compute the $(4/3-\varepsilon)$-approximate diameter
	of an undirected unweighted dense graph with truly
	subcubic preprocessing time and truly subquadratic update/query time.
	This result is surprising since in the static setting it is not clear
	whether a reduction from BMM to diameter is possible.
	We further show under the BMM conjecture that many problems, such as reachability or
	approximate shortest paths, cannot be solved faster than by recomputation from scratch
	even after \emph{only one or two} edge insertions. We extend our reduction from BMM to Diameter to give a reduction from All Pairs Shortest Paths to Diameter under one deletion in weighted graphs. This is intriguing, as in the static setting it is a big open problem whether Diameter is as hard as APSP. We further get a nearly tight lower bound for shortest paths after two edge deletions based on the APSP conjecture. We give more lower bounds under the
	Strong Exponential Time Hypothesis. Many of our lower bounds also hold for static oracle data structures where no sensitivity is required.
Finally, we give the first algorithm for the $(1+\varepsilon)$-approximate
	radius, diameter, and eccentricity problems in directed or undirected
	unweighted graphs in case of single edges failures.
	The algorithm has a truly subcubic running time for graphs with a truly subquadratic number of edges;
	it is tight w.r.t.\ the conditional lower bounds we obtain.
\end{abstract}

%
\section{Introduction}
\label{Sec:Introduction}

A dynamic algorithm is an algorithm that is able to handle changes in the input data:
It is given an input instance $x$ and is required to
maintain certain properties of $x$ while $x$ undergoes (possibly very many) updates.
For example, an algorithm might maintain a graph, which undergoes edge insertions and deletions,
and a query is supposed to the return the diameter of the graph after the updates.
Often dynamic algorithms are also referred to as data structures.
During the last few years strong conditional lower bounds for many dynamic problems were derived
(see, e.g., \cite{Patrascu10,abboud2014popular,henzinger2015unifying,abboud2015matching,dahlgaard2016hardness,abboud2016popular,kopelowitz2016higher}),
which rule out better algorithms than simple recomputation from scratch after each update or before each query.

Partially due to this, in recent years it has become popular to study dynamic problems in a
more restricted setting that only allows for a \emph{bounded} number of
changes to the input instance (see, for example,
\cite{patrascu2007planning,duan2010connectivity,bernstein2009nearly,chechik2012sensitivity},
and the references in Table~\ref{Tbl:UpperBounds}).
These algorithms are usually referred to as \emph{sensitivity}\footnote{Sometimes sensitivity data structures are also
called ``fault-tolerant'' or ``emergency planning'' algorithms.
See Appendix~\ref{Sec:Terminology} for a discussion of
terminology.} data structures.
The hope is to obtain algorithms in the sensitivity setting which
are faster than the conditional lower bounds for the general setting.

More formally, a \emph{data structure with sensitivity~$d$} for
a problem~$P$ has the following properties:
It obtains an instance $p$ of $P$ and is allowed polynomial preprocessing time on $p$.
After the preprocessing, the data structure must provide the following operations:

 {\bf  (Batch) Update:} Up to $d$ changes are performed to the {\em initial} problem instance $p$,
  				e.g., $d$ edges are added to or removed from $p$.
 
 {\bf Query:} The user queries a specific property about the instance of the problem
  				after the last update, e.g., the shortest path between two nodes avoiding the
				edges deleted in the last update.

The parameter $d$ bounding the batch update size is referred to as the \emph{sensitivity}
of the data structure. 
Note that every batch update is performed on the {\em original} problem instance.

Thus, in contrast to ``classic'' dynamic algorithms (without sensitivity),
a query only reflects the changes made to $p$ by the \emph{last} batch update and
{\em not} by previous batch updates. As  the size of a batch update is constrained to at most $d$,
each query is executed on a graph that differs from $p$ by at most $d$ edges.
After a batch update an arbitrary number of queries may be performed.

Some data structures (usually called \emph{oracles}) combine a query and an update
into a single operation, i.e., the combined operation obtains an input tuple $(Q,U)$, where
$Q$ is a query and $U$ is an update. A special case are {\em static oracles}, which have $U = \emptyset$.
The conditional lower bounds we derive in this paper also hold in this setting,
since oracles with an efficient combined operation can be used to solve
sensitivity problems.

While some existing sensitivity data structures can preprocess the answers to all possible updates and queries
during their preprocessing time, this is not possible in general (due to constraints in the preprocessing time
and the fact that the number of possible updates/queries grows exponentially in the parameter $d$).
Hence, we still consider a sensitivity data structure a dynamic (instead of static) algorithm.

\paragraph*{The Hypotheses.}
We state the hypotheses on which we base the conditional lower bounds in this paper. 
By now they are all considered standard in proving fine-grained reduction-based lower bounds. For a more detailed description of the hypotheses, see, e.g.,
Abboud and Williams~\cite{abboud2014popular}, Henzinger et al.~\cite{henzinger2015unifying},
and the references therein.
As usual we work in the word-RAM model of computation with word length of $O(\log n)$ bits. The hypotheses below concern the complexity of the Boolean Matrix Multiplication (BMM), Satisfiability of Boolean Formulas in Conjunctive Normal Form (CNF-SAT), All Pairs Shortest Paths (APSP), Triangle Detection and Online Boolean Matrix Vector Multiplication (OMv) problems. Other popular hypotheses from prior work consider other famous problems such as $3$SUM and other sparsity regimes such as triangle detection in very sparse graphs (see, e.g. \cite{abboud2014popular}).


\begin{conjecture}[Impagliazzo, Paturi and Zane \cite{impagliazzo2001complexity,impagliazzo2001which}]
	The Strong Exponential Time Hypothesis (SETH) states that for each $\varepsilon > 0$,
	there exists a $k \in \mathbb{N}$, such that $k$-SAT cannot be solved in time $O(2^{n(1-\varepsilon)} \poly(n))$.
\end{conjecture}


\begin{conjecture}
	The Boolean Matrix Multiplication (BMM) conjecture states that for all $\varepsilon > 0$,
	there exists no \emph{combinatorial} algorithm that computes the product
	of two $n \times n$ matrices in expected time $O(n^{3 - \varepsilon})$.
\end{conjecture}

Note that BMM can be solved in truly subcubic using fast matrix multiplication (FMM): the current fastest algorithms run in $O(n^{2.373})$ time~\cite{v12,legallmult}. However,
algorithms using FMM are not considered to be combinatorial.
Formally, the term \emph{combinatorial} algorithm is not well-defined and it is common to
rule out the use of FMM or other ``Strassen-like'' methods in the design of such algorithms as most of them are not considered practical. True combinatorial algorithms are not only considered practical but also easily extendable. For instance, prior work on combinatorial BMM algorithms has almost always led to an algorithm for APSP with similar running time (e.g. \cite{fourrus} and \cite{chan07}).

One of the simplest graph problems is that of detecting whether the graph contains a triangle, i.e., three nodes with all three edges between them. Itai and Rodeh~\cite{itairodeh} showed that any algorithm for BMM can solve Triangle detection in the same time. Conversely, Vassilevska Williams and Williams~\cite{williams2010subcubic} showed that any
truly subcubic combinatorial algorithm for Triangle Detection can be converted into a
truly subcubic combinatorial algorithm for BMM. Hence, the BMM conjecture implies there is no
truly subcubic \emph{combinatorial} algorithm for Triangle Detection. We use this fact and the resulting Triangle Conjecture that there is no truly subcubic algorithm for Triangle Detection in our reductions based on BMM.

The following is a popular conjecture about the APSP problem.

\begin{conjecture}
	The APSP conjecture states that given a graph $G$ with $n$ vertices, $m$ edges, and
	edge weights in $\{1,\dots,n^c\}$ for some constant $c$,
	the All Pairs Shortest Paths problem (APSP) cannot be solved in $O(n^{3-\varepsilon})$ expected time
	for any $\varepsilon > 0$.
\end{conjecture}

Similar to the relationship between BMM and Triangle Detection, \cite{williams2010subcubic} showed that there is a triangle problem in weighted graphs, Negative Triangle, that is equivalent under subcubic reductions to APSP. We use that problem in our reductions.

Our final conjecture concerns the online version of Boolean matrix vector product.

\begin{conjecture}[Henzinger et al.~\cite{henzinger2015unifying}]
	Let $B$ be a Boolean matrix of size $n \times n$.
	In the Online Matrix-vector (OMv) problem, $n$ binary vectors $v_1, \dots, v_n$ of size $n$
	appear online and an algorithm solving the problem must output the vector $B v_i$ before
	the next vector $v_{i+1}$ arrives.

	The OMv conjecture states that for all $\varepsilon > 0$ and after any polynomial time preprocessing of $B$,
	it takes $\Omega( n^{3-\varepsilon} )$ time to solve the OMv problem with error probability at most $1/3$.
\end{conjecture}

Most of the conjectures are stated w.r.t.\ \emph{expected} time, i.e., the
conjectures rule out randomized algorithms.
In case of dynamic algorithms using randomness, it is common to argue if an oblivious
or a non-oblivious adversary is allowed.
Previous literature on conditional lower bounds for dynamic algorithms
did not explicitly state what kind of adversaries are allowed for their lower bounds.
We give a quick discussion of this topic in Appendix~\ref{Sec:Adversaries}.

\paragraph*{Our Results.}
In this paper we develop a better understanding of the possibilities and limitations
of the sensitivity setting by providing conditional lower bounds for sensitivity problems.
We show that under plausible assumptions for many dynamic graph problems even
the update of only \emph{one or two} edges cannot be solved faster than by re-computation from 
scratch. See Table~\ref{Tbl:LowerBounds} 
and Table~\ref{Tbl:StaticLowerBounds} 
in the Appendix  for a list of all our conditional lower bounds 
for sensitivity data structures, and our lower bounds for static oracles respectively. Table~\ref{Tbl:Problems} gives explanations of the problems.
The abbreviations used in the tables are explained in its captions. We next discuss our main results.

\paragraph*{New reductions.}
We give several new reductions for data structures with small
sensitivity. 

{\bf (1)} We give a novel reduction from triangle detection and BMM
to maintaining an approximation of the diameter of the graph and
eccentricities of all vertices, under a single edge failure. This is particularly surprising because in the static
case it is unknown how to reduce BMM to diameter computation.
Using the BMM conjecture this results in lower bounds 
of $n^{3-o(1)}$ on the preprocessing time
or of $n^{2-o(1)}$ update or query time
for $(4/3-\varepsilon)$-approximate decremental diameter
and eccentricity in unweighted graphs {\em with sensitivity 1}, i.e., when a single edge is deleted.
Those results are tight w.r.t.\ the algorithm we present in Section~\ref{Sec:DiameterUpperBound}.

{\bf (2)} A particular strength of BMM-based reductions is that they can very often be converted into APSP-based lower bounds
for weighted variants of the problems. APSP-based lower bounds, in turn, no longer require the
``combinatorial''-condition on the algorithms, making the lower bounds stronger.
We show how our BMM-based lower bounds for approximate diameter with sensitivity $1$ can be converted into an
APSP-based lower bound for diameter with sensitivity $1$ in weighted graphs.
In particular, we show that unless APSP has a truly subcubic algorithm, any data structure
that can support diameter queries for a single edge deletion must either have essentially
cubic preprocessing time, or essentially quadratic query time.
This lower bound is tight w.r.t.\ to a trivial algorithm using the data structure of~\cite{bernstein2009nearly}.
The APSP to $1$-sensitive Diameter lower bound is significant also because it is a big open problem whether in the static case Diameter and APSP are actually subcubically equivalent (see e.g.~\cite{AbboudGW15}).

{\bf (3)} We consider the problem of maintaining the distance between two fixed nodes $s$ and $t$ in an undirected weighted graph under edge failures.
The case of a {\em single edge failure} can be solved in $m$ edge, $n$ node graphs with essentially optimal $O(m\alpha(n))$ preprocessing time and $O(1)$ query time with an algorithm of Nardelli et al.~\cite{nardelli}.
The case of two edge failures has been open for some time. We give compelling reasons for this by showing that under the APSP conjecture, maintaining the $s$-$t$ distance in an unweighted graph under two edge failures requires either $n^{3-o(1)}$ preprocessing time or $n^{2-o(1)}$ query time. Notice that with no preprocessing time, just by using Dijkstra's algorithm at query time, one can obtain $O(n^2)$ query time. Similarly, one can achieve $\tilde{O}(n^3)$ preprocessing time and $O(1)$ query time by applying the single edge failure algorithm of~\cite{nardelli} $n$ times at preprocessing, once for $G\setminus \{e\}$ for every $e$ on the shortest $st$ path. Thus our lower bound shows that under the APSP conjecture, the naive recomputation time is essentially optimal.

{\bf (4)}  We show lower bounds with sensitivity $d$ for deletions-only and insertions-only
$(2-\varepsilon)$-approximate single source and
$(5/3 - \varepsilon)$-approximate $st$-shortest paths in undirected unweighted graphs,
as well as for weighted bipartite matching problems under the OMv conjecture. The 
lower bounds show that with polynomial in $n$ preprocessing
either the update time must be super-polynomial in $d$ or the query time 
must be $d^{1-o(1)}$.

\paragraph*{New upper bounds.}
We complement our lower bounds with an algorithm showing that some of our lower bounds are tight:
In particular, we present a deterministic combinatorial algorithm that can
compute a $(1+\varepsilon)$-approximation (for any $\varepsilon>0$) for the eccentricity of any given vertex,
the radius and the diameter of a directed or undirected unweighted graph after single edge failures.
The preprocessing time of the data structure is $\tO(mn+n^{1.5} \sqrt{D m / \varepsilon})$, where $D$ is the diameter of the graph and $m$ and $n$ are the number of edges and vertices; the query time is constant.
Since $D \leq n$, the data structure can
be preprocessed in time $\tO(n^2 \sqrt{m / \varepsilon})$. In particular, for sparse graphs with
$m = \tO(n)$, it takes time $\tO(n^{2.5} \varepsilon^{-\frac{1}{2}})$ to build
the data structure. Our lower bounds from BMM state that even getting
a $(4/3-\varepsilon)$-approximation for diameter or eccentricity after a \emph{single} edge deletion
requires either $n^{3-o(1)}$ preprocessing time, or $n^{2-o(1)}$ query or update time.
Hence, our algorithm's preprocessing time is tight (under the conjecture)
since it has constant time queries.

\paragraph*{Conditional Lower Bounds based on modifications of prior reductions.}
Some reductions in prior work~\cite{williams2010subcubic,abboud2014popular} only
perform very few updates before a query is performed or they can be modified to do so. After the query, the updates are ``undone''
by rolling back to the initial instance of the input problem.
Hence, some of their reductions also hold in a sensitivity setting.
Specifically we achieve the following results in this way:

(1) Based on the BMM conjecture we show that for reachability problems with $st$-queries already \emph{two} edge
insertions require $n^{3-o(1)}$ preprocessing time
or $n^{2-o(1)}$ update or query time;
for $ss$-queries we obtain the same bounds even for a \emph{single} edge insertion.
This lower bound is matched by an algorithm that recomputes at each step.

(2) We present strong conditional lower bounds for static oracle data structures.
We show that under the BMM conjecture, oracle data structures that answer about the reachability between any two queried vertices 
cannot
have truly subcubic preprocessing time \emph{and} truly subquadratic query time.
This implies that combinatorial algorithms  \emph{either} essentially need to compute the transitive closure
matrix of the graph during the preprocessing time \emph{or} essentially need to traverse the graph at each query.
We show the same lower bounds for static oracles that solve the $(5/3-\varepsilon)$-approximate
$ap$-shortest paths problem in undirected unweighted graphs.
This shows that we essentially cannot do better than solving APSP in the preprocessing or computing the distance
in each query.

(3) The subcubic equivalence between the replacement paths problem and APSP~\cite{williams2010subcubic} immediately leads to a conditional lower bound 
for $s$-$t$ distance queries with sensitivity 1 in {\em directed}, weighted graphs. Our lower bound for $s$-$t$ distance queries with sensitivity 2 in undirected graphs is inspired by this reduction.
The lower bound for sensitivity $1$ is matched by the algorithm of Bernstein and Karger~\cite{bernstein2009nearly}.

Similarly, a reduction from BMM to replacement paths in directed unweighted graphs from~\cite{williams2010subcubic} shows that the $O(m\sqrt{n})$ time algorithm of Roditty and Zwick~\cite{rzkshortest} is optimal among all combinatorial algorithms, for every choice of $m$ as a function of $n$. It also immediately implies that under the BMM conjecture, combinatorial  $s$-$t$ distance $1$-sensitivity oracles in unweighted graphs require either $mn^{0.5-o(1)}$ preprocessing time or $m/n^{0.5+o(1)}$ query time, for every choice of $m$ as a function of $n$; this is tight due to Roditty and Zwick's algorithm.
(The combinatorial restriction is important here as there is a faster $\tilde{O}(n^{2.373})$ time non-combinatorial algorithm for replacement paths \cite{v-replacement} and hence for distance sensitivity oracles in directed unweighted graphs.)

(4) We additionally provide new lower bounds under SETH:
We show that assuming SETH the \#SSR problem cannot be solved with truly subquadratic
update and query times when any constant number of edge insertions is allowed;
this matches the lower bound for the general dynamic setting.
For the $ST$-reachability problem and the computation of
$(4/3 - \varepsilon)$-approximate diameter we show that under SETH truly sublinear
update and query times are not possible even when only a constant number of edge insertions are supported.
The sensitivity of the reductions is a constant $K(\varepsilon,t)$ that is
determined by the preprocessing time $O(n^t)$ we allow
and some properties of the sparsification lemma~\cite{impagliazzo2001which}.
Notice that while the constant $K(\varepsilon,t)$ depends on the preprocessing time and the constant
in the sparsification lemma, it does \emph{not} depend on any property of the SAT instance in the reduction.
See Section~\ref{Sec:SETHSensitivity} for a thorough discussion of the parameter $K(\varepsilon,t)$.
The lower bound for \#SSR shows that we cannot do better than recomputation after each update.

(5) Using a reduction from OMv we show lower bounds with sensitivity $d$ for deletions-only or insertions-only
$st$-reachability, strong connectivity in directed graphs. The 
lower bounds show that with polynomial in $n$ preprocessing either the update time must be super-polynomial $d$ or the query time 
must be $\Omega(d^{1-\varepsilon})$.

\paragraph*{Related Work.}
In the last few years many conditional lower bounds were derived for dynamic algorithms.
Abboud and Williams~\cite{abboud2014popular} gave such lower bounds under several different conjectures.
New lower bounds were given by Henzinger et al.~\cite{henzinger2015unifying}, who introduced the OMv conjecture,
and by Abboud, Williams and Yu~\cite{abboud2015matching}, who stated combined conjectures
that hold as long as either the 3SUM conjecture \emph{or} SETH \emph{or} the APSP conjecture is correct.
Dahlgaard~\cite{dahlgaard2016hardness} gave novel lower bounds for partially dynamic algorithms.
Abboud and Dahlgaard~\cite{abboud2016popular} showed the first hardness results for dynamic
algorithms on planar graphs and
Kopelowitz, Pettie and Porat~\cite{kopelowitz2016higher} gave stronger lower bounds from the 3SUM conjecture.
However, none of the lower bounds mentioned in the above papers explicitly handled the sensitivity setting.

During the last decade there have been many new algorithms designed for the sensitivity setting.
In Section~\ref{Sec:UpperBounds} we give a short discussion summarizing many existing algorithms.

\section{Lower Bounds From Boolean Matrix Multiplication}

The following theorem summarizes the lower bounds we derived from the BMM conjecture.
\begin{theorem}
\label{Thm:TriangleLBs}
	Assuming the BMM conjecture, combinatorial algorithms cannot solve
	the following problems with preprocessing time $O(n^{3-\varepsilon})$,
	and update and query times $O(n^{2-\varepsilon})$ for any $\varepsilon > 0$:
	\begin{compactenum}
		\item incremental $st$-reachability with sensitivity 2,
		\item incremental $ss$-reachability with sensitivity 1,
		\item static $ap$-reachability,
		\item $(7/5-\varepsilon)$-approximate $st$ shortest paths in undirected unweighted graphs with sensitivity 2,
		\item $(3/2-\varepsilon)$-approximate $ss$ shortest paths in undirected unweighted graphs with sensitivity 1,
		\item static $(5/3-\varepsilon)$-approximate $ap$ shortest paths 
		\item decremental $(4/3 - \varepsilon)$-approx.\ diameter in undirected unweighted graphs with sensitivity 1,
		\item decremental $(4/3 - \varepsilon)$-approx.\ eccentricity in undirected unweighted graphs with sensitivity 1.
	\end{compactenum}
	Additionally, under the BMM conjecture, decremental $st$-shortest paths with sensitivity $1$ in directed unweighted graphs with $n$ vertices and $m\geq n$ edges require either $m^{1-o(1)}\sqrt{n}$ preprocessing time or $m^{1-o(1)}/\sqrt{n}$ query time for every function $m$ of $n$.
\end{theorem}

A strength of the reductions from BMM is that they
can usually be extended to provide APSP-based reductions
for weighted problems without the restriction to combinatorial algorithms;
we do this in Section~\ref{Sec:APSP}.
While we state our results in the theorem only for combinatorial algorithms
under the BMM conjecture, we would like to point out that they also hold {\em for any kind of algorithm}
under a popular version of the triangle detection conjecture for sparse graphs that states that finding a triangle in an $m$-edge graph requires $m^{1+\delta-o(1)}$ time for some $\delta>0$.
Our lower bounds then rule out algorithms with a
preprocessing time of $O(m^{1+\delta-\varepsilon})$ and update and query times
$O(m^{2\delta - \varepsilon})$ for any $\varepsilon > 0$.

Many of the bullets of the theorem follow from prior work via a few observations,
which we discuss in Appendix~\ref{subsec:ReachSPAppendix}.
Our results on decremental diameter and eccentricity, however, are completely novel.
In fact, it was completely unclear before this work whether such results are possible.
Impagliazzo et al.~\cite{carmosino2016nondeterministic} define a strengthening of SETH under which
there can be no deterministic fine-grained reduction from problems such as APSP and BMM
to problems such as orthogonal vectors or diameter in sparse graphs.
It is not clear whether a reduction from BMM to diameter in dense graphs is
possible, as the same ``quantifier issues'' that arise in the sparse graph case
arise in the dense graph case as well: Diameter is an $\exists\forall$-type
problem (i.e., do there exist two nodes such that all paths between them are long?),
and BMM is equivalent to Triangle detection which is an $\exists$-type problem (i.e., do there exist three nodes that form a clique?).

\paragraph*{Decremental Diameter.}
We give the reduction from BMM to decremental diameter in
undirected unweighted graphs with sensitivity $1$.
Note that the lower bound also holds for eccentricity oracles: Instead of querying the
diameter $n$ times, we can query the eccentricity of a variable vertex $n$ times. 

	Let $G = (V,E)$ be an undirected unweighted graph for Triangle Detection. We construct
	a graph $G'$ as follows.
	
	We create four copies of $V$ denoted by $V_1, V_2, V_3, V_4$,
	and for $i = 1,2,3$, we add edges
	between nodes $u_i \in V_i$ and $v_{i+1} \in V_{i+1}$ if $(u,v) \in E$.
	We create vertices $a_v$ and $b_v$ for each $v \in V$,
	and denote the set of all $a_v$ by $A$ and the set of all $b_v$ by $B$.
	We connect the vertices in $A$ to a clique and also those of $B$.
	For each $v \in V$, we add an edge $(v_1, a_v)$ and
	an edge $(a_v, b_v)$.
	A node $b_v$ is connected to all vertices in $V_4$.
	We further introduce two additional vertices $c, d$, which are
	connected by an edge.
	We add edges between $c$ and all nodes in $V_2$ and $V_3$,
	and between $d$ and all nodes in $V_3$ and $V_4$.
	The node $c$ has an edge to each vertex
	in $A$ and the node $d$ has an edge to each vertex in $B$.
	Notice that the resulting graph has $O(n)$ vertices and
	$O(n^2)$ edges. We visualized the graph in Figure~\ref{fig:DiamPic} in the appendix.

	Note that even without the edges from $B \times V_4$, no pair of nodes has distance larger than $3$,
	except for pairs of nodes from $V_1 \times V_4$.
	If a node $v$ participates in a triangle in $G$, then in $G'$ there is a path of length $3$ from $v_1$ to $v_4$
	without an edge from $B \times V_4$. Otherwise, there is no such path, i.e., the diameter
	increases to $4$ after the deletion of $(b_v, v_4)$.

	We perform one stage per vertex $v \in V$: Consider the copy $v_4 \in V_4$ of $v$.
	We remove the edge $(b_v,v_4)$ and query the diameter of the graph.
	We claim that $G$ has a triangle iff one of the queries returns diameter $3$.
	
	\begin{lemma}
		For each vertex $v$ in $G$, the diameter of $G' \setminus \{ (b_v, v_4) \}$ is larger than 3
		if and only if $v$ does not participate in a triangle in $G$.
	\end{lemma}
	\begin{proof}
		Assume that $G$ has a triangle $(v,u,w) \in V^3$ and consider the stage for $v$.
		Notice that only the shortest paths change that used edge $(b_v,v_4)$;
		this is not the case for any $z \neq v$, because the path $z_1 \to a_z \to b_z \to z_4$
		is not affected by the edge deletion.
		We only need to consider the path $v_1 \to a_v \to b_v \to v_4$.
		Since $G$ has a triangle $(v,u,w)$, there exists the path $v_1 \to u_2 \to w_3 \to v_4$
		of length $3$ as desired.
		Hence, the diameter is $3$.
		
		Assume the query in the stage for vertex $v \in V$ returned diameter $3$.
		Since we deleted the edge $(b_v,v_4)$, there is no path of length $3$
		from $v_1$ to $v_4$ via $A$ and $B$.
		Hence, the new shortest path from $v_1$ to $v_4$ must have the form $v_1 \to u_2 \to w_3 \to v_4$.
		By construction of the graph, this implies that $G$ has a triangle $(v,u,w)$.
	\end{proof}
	
	Altogether we perform $n$ queries and $n$ updates. Thus under the BMM conjecture
	any combinatorial algorithm requires $n^{3-o(1)}$ preprocessing time or
	$n^{2 - o(1)}$ update or query time.

\section{Sensitivity Lower Bounds from the APSP Conjecture}
\label{Sec:APSP}
In this section we present new lower bounds based on the APSP conjecture.
These lower bounds hold for arbitrary, not necessarily combinatorial, algorithms.
We present our results in the following theorem and give the proofs in Appendix \ref{subsec:APSPproofsApendix}.

\begin{theorem}
\label{Thm:APSPLBs}
	Assuming the APSP conjecture, no algorithms can solve
	the following problems with preprocessing time $O(n^{3-\varepsilon})$,
	and update and query times $O(n^{2-\varepsilon})$ for any $\varepsilon > 0$:
	\begin{compactenum}
		\item Decremental $st$-shortest paths in directed weighted graphs with sensitivity $1$,
		\item decremental $st$-shortest paths in undirected weighted with sensitivity $2$,
		\item decremental diameter in undirected weighted graphs with sensitivity $1$.
	\end{compactenum}
\end{theorem}

\paragraph*{Decremental $st$-shortest paths in directed weighted graphs with sensitivity $1$.}
In 2010, Vassilevska Williams and Williams~\cite{williams2010subcubic} showed that the so called Replacement Paths (RP) problem is subcubically equivalent to APSP. RP is defined as follows: given a directed weighted graph $G$ and two nodes $s$ and $t$, compute for every edge $e$ in $G$, the distance between $s$ and $t$ in $G\setminus \{e\}$.
Note that only the deletion of the at most $n - 1$ edges on the shortest path from $s$ to $t$
affect the distance from $s$ to $t$.
This has an immediate implication for $1$-sensitivity oracles for $st$-shortest paths:
The APSP conjecture would be violated by any $1$-sensitivity oracle that uses $O(n^{3-\varepsilon})$ preprocessing time
and can answer distance queries between two fixed nodes $s$ and $t$ with one edge deletion in time $O(n^{2-\varepsilon})$ for any $\varepsilon>0$. 

\paragraph*{Decremental $st$-shortest paths in undirected weighted with sensitivity $2$.}
\label{subsubsec:APSPshortpathsDec2}
With this we are able to
show that on \emph{undirected} weighted graphs finding a shortest path between
fixed $s$ and $t$ with $2$ edge deletions cannot be done with truly sub-cubic preprocessing time and truly subquadratic query time assuming the APSP conjecture. 
This is surprising because in the case of a {\em single edge failure} Nardelli et al.~\cite{nardelli} show that shortest paths can be solved with an essentially optimal $O(m\alpha(n))$ preprocessing time and $O(1)$ query time. Thus, assuming the APSP conjecture we show a separation between $1$ sensitivity and $2$ sensitivity. Additionally, with sensitivity $2$ and no preprocessing time $O(n^2)$ update time is achievable, and with $\tilde{O}(n^3)$ preprocessing time using Nardelli et al. we can get an $O(1)$ query time. Thus, we show these approaches are essentially tight assuming the APSP conjecture. 
The full reduction is in Appendix~\ref{subsubsec:APSPshortpathsDec2}.

\paragraph*{Decremental diameter in undirected weighted graphs with sensitivity $1$.}
\label{subsec:APSPdecDiam}
A nice property of BMM-based reductions is that they can very often be converted to APSP-based reductions
to weighted versions of problems. Here we convert our BMM-based reduction for decremental 1-sensitive Diameter
to a reduction from APSP into decremental 1-sensitive diameter in undirected weighted graphs.
Note that, as in the BMM case we can get the same lower bounds for eccentricity.

\section{SETH Lower Bounds with Constant Sensitivity}
\label{Sec:NewSETHLowerBounds}

In this section, we prove conditional lower bounds with constant sensitivity from SETH.
Before we give the reductions, we first argue about what their sensitivities are.

\begin{theorem}
  Let $\varepsilon > 0$, $t \in \mathbb{N}$.
  The SETH implies that there exists no algorithm with preprocessing time $O(n^t)$,
  update time $u(n)$ and query time $q(n)$, such that 
  $\max\{ u(n), q(n) \}= O(n^{1-\varepsilon})$ for the following problems:
  \begin{compactenum}
		\item Incremental \#SSR with constant sensitivity $K(\varepsilon,t)$,
		\item $(4/3 - \varepsilon)$-approximate incremental diameter with constant sensitivity $K(\varepsilon, t)$,
		\item incremental ST-Reach with constant sensitivity $K(\varepsilon, t)$.
  \end{compactenum}
\end{theorem}
We prove the theorem in Appendix \ref{subsec:SETHAppendix}.
The parameter $K(\varepsilon,t)$ is explained in the following paragraph*.

\paragraph*{The Sensitivity of the Reductions.}
\label{Sec:SETHSensitivity}
The conditional lower bounds we prove from SETH hold even for constant sensitivity;
however, the derivation of these constants is somewhat unnatural.
Nonetheless, we stress that our lower bounds hold for constant sensitivity
and in particular for every algorithm with sensitivity $\omega(1)$.

To derive the sensitivity of our reductions, we use a similar approach as
Proposition~1 in \cite{abboud2014popular}, but we need
a few more details. We start by revisiting the sparsification lemma.

\begin{lemma}[Sparsification Lemma, \cite{impagliazzo2001which}]
  For $\varepsilon > 0$ and $k \in \mathbb{N}$, there exists a constant $C = C(\varepsilon, k)$,
  such that any $k$-SAT formula $F$ with $\tilde n$ variables can be expressed as
  $F = \bigvee_{i=1}^l F_i$, where $l = O(2^{\varepsilon \tilde n})$ and each $F_i$ is a $k$-SAT
  formula with at most $C \tilde n$ clauses.
  This disjunction can be computed in time $O(2^{\varepsilon \tilde n} \poly(\tilde n))$.
\end{lemma}

We set $C(\varepsilon)$ to the smallest $C(\varepsilon,k)$, over all $k$ such that $k$-SAT cannot be solved faster
than in $O^*(2^{(1-\varepsilon)\tilde n})$ time\footnote{The $O^*(\cdot)$ notation hides $\poly(\tilde n)$ factors.};
formally, 
	$C(\varepsilon) = \min\{ C(\varepsilon', k) :
	\varepsilon' < \varepsilon \text{ and } k \in \mathbb{N} \text{ with } k\text{-SAT} \not\in O^*(2^{(1-\varepsilon')\tilde n})\}$.
Note that $C(\varepsilon)$ is well-defined if we assume that SETH is true
(see also Proposition~1 in \cite{abboud2014popular}).
Finally, for any $\varepsilon > 0$ and $t \in \mathbb{N}$, we set
	$K(\varepsilon, t) = C(\varepsilon) \cdot t / (1 - \varepsilon)$,
which gives the sensitivity of our reductions.

In our reductions, $t \in \mathbb{N}$ is the exponent of the allowed preprocessing time and $\varepsilon > 0$
denotes the improvement in the exponent of the running time over the $2^n$-time algorithm.
We note that $K(\varepsilon,t)$ gives a tradeoff: For small $t$ (i.e., less preprocessing time),
the lower bounds hold for smaller sensitivities;
a smaller choice of $\varepsilon$ yields larger sensitivities.

In the reductions we will write $K$ to denote $K(\varepsilon, t)$ and
$c$ to denote $C(\varepsilon,k)$ whenever it is clear from the context.

\paragraph*{The Reductions.}
Our reductions are conceptually similar to the ones in~\cite{abboud2014popular}, but the
graph instances we construct are based on a novel idea to minimize the size
of the batch updates we need to perform.
Here we describe the construction of the graphs we use in the reductions and
refer to Appendix~\ref{subsec:SETHAppendix} for full proofs.

We give two graphs, $H_\delta$ and $D_\delta$, for $\delta \in (0,1)$.
For the construction,
let $F$ be a SAT formula over a set $V$ of $\tilde n$ variables and $c \cdot \tilde n$ clauses.
Let $U \subset V$ be a subset of $\delta \tilde n$ variables.

Construction of $H_\delta$:
  For each partial assignment to the variables in $U$ we
  introduce a node. The set of these nodes is denoted by $\bar U$.
  For each clause of $F$ we introduce a node and denote the set of these nodes by $C$.
  We add an edge between a partial assignment $\bar u \in \bar U$ and
  a clause $c \in C$ if $\bar u$ does not satisfy $c$.
  Observe that $H_\delta$ has $O(2^{\delta \tilde n})$ vertices and $O^*(2^{\delta \tilde n})$ edges.

Construction of $D_\delta$:
  We partition the set of clauses $C$ into $K = c/\delta$ groups
  of size $\delta \tilde n$ each and denote these groups by $G_1, \dots, G_{K}$.
  For all groups $G_i$, we introduce a vertex into the graph
  for each non-empty subset $g$ of $G_i$.
  The edges to and from the nodes of $D_\delta$ will be introduced during reductions.
  Observe that for each group we introduce $O(2^{\delta \tilde n})$ vertices and
  $D_\delta$ has $O(K \cdot 2^{\delta \tilde n})$ vertices in total.

Our reductions have small sensitivity since we will only need to insert a single edge
from $H_\delta$ to each group of clauses in $D_\delta$.
Hence, we only need to insert $K = O(1)$ edges
in order to connect $H_\delta$ and $D_\delta$ at each stage in the reduction.
However, we will need to argue how we can efficiently pick the correct sets in $D_\delta$.

\section{Diameter Upper Bound}
\label{Sec:DiameterUpperBound}

In this section, we present deterministic algorithms, which can
compute a $(1+\varepsilon)$-approximation for the eccentricity,
the radius and the diameter of directed and undirected unweighted graphs after single edge deletions.
All of these algorithms run in time truly subcubic time for graphs with
a truly subquadratic number of edges.

Bernstein and Karger~\cite{bernstein2009nearly} give an algorithm for the related problem
of all-pairs shortest paths in a directed weighted graph $G = (V,E)$ in case of
single edge deletions.
Their oracle data structure requires $\tO(mn)$ preprocessing time. Given a
triplet $(u,v,e) \in V^2 \times E$, the oracle can output the distance from $u$
to $v$ in $G \setminus e$ in $O(1)$ time.

For the diameter problem with single edge deletions, note that
only deletions of the edges in the shortest paths trees can have an effect
on the diameter.
Using this property, a trivial algorithm to compute the exact diameter after the deletion of
a single edge works as follows: Build the oracle data structure
of Bernstein and Karger~\cite{bernstein2009nearly}.
For each vertex $v$, consider its shortest paths tree $T_v$.
Delete each tree-edge once and query the distance from $v$ to $u$ in $G \setminus e$
for all vertices $u$ in the subtree of the deleted tree-edge.
By keeping track of the maximum distances, the diameter of $G$ after
a single edge deletion is computed exactly.
As there are $n - 1$ edges in $T_v$, we spend $O(n^2)$ time for each vertex.
Thus, the trivial algorithm requires $O(n^3)$ time.

In this section, we improve upon this result as follows.

\begin{theorem}
	Let $G = (V,E)$ be a directed or undirected unweighted graph with $n$ vertices and $m$ edges, let
	$\varepsilon > 0$, and let $D$ be the diameter of $G$.
	There exists a data structure that given
	a single edge $e \in E$ returns for $G \setminus e$ in constant time
	(1) the diameter, (2) the radius,
	and (3) the eccentricity of any vertex $v \in V$
	within an approximation ratio of $1 + \varepsilon$.
	It takes $\tO(n^{1.5} \sqrt{D m / \varepsilon}+mn)$ preprocessing time to build this data structure.
\end{theorem}

The rest of this section is devoted to the proof of the theorem.
We give the proof of the theorem for directed graphs and point out the same proof also works
for undirected graphs.
We first describe how we can answer queries for the eccentricity
of a fixed vertex $v \in V$ after a single edge deletion.
After this, we explain how to extend this algorithm
to solve the diameter and the radius problems after single edge deletions,
and analyse the correctness and running time of the algorithm.

The data structure preprocesses the answers to all queries.
Then queries can be answered via table lookup in $O(1)$ time.

\paragraph*{Preliminaries.}
Let $G = (V,E)$ be a directed unweighted graph.
For two vertices $u, u' \in V$ we denote the distance of
$u$ and $u'$ in $G$ by $\dG(u,u')$. For an edge $e \in E$ and vertices $u, u' \in V$,
we denote the distance in the graph $G \setminus e$ by $\de(u,u')$.
Given a tree $T$ with root $v$ and a tree-edge $e \in T$, we
denote the subtree of $T$ that is created when $e$ is removed from $T$
and that does \emph{not} contain $v$ by $T_e$. We let $d_e$ be the height of $T_e$.
A node $u$ in $T$ has \emph{level} $i$, if $\dG(v,u) = i$.

Let $F \in \mathbb{N}$ be some suitably chosen parameter
(see the last paragraph* of this section).
Then given a tree $T$ with root $v$, we call a tree-edge $e$
\emph{high}, if both of its endpoints have level less than $F$ from $v$;
we call all other edges \emph{low}.
We denote the set of all high edges by $T_<$, i.e.,
$T_< = \{ e = (w,w') \in T : \dG(v,w) < F, \dG(v,w') < F\}$;
the set of all low edges is given by $T_>$.

\paragraph*{The Algorithm.}
Our data structure preprocesses the answers to all queries,
and then queries can be answered via table lookup in $O(1)$ time.
The preprocessing has three steps: First, in the initialization phase, we compute
several subsets of vertices that are required in the next steps.
Second, we compute the eccentricity of $v$ after the deletion of a high edge exactly.
We compute it exactly, since after the deletion of a tree-edge high up in the shortest path tree $T_v$ of $v$,
the nodes close to $v$ in $T_v$ might ``fall down'' a lot.
This possibly affects all vertices in the corresponding subtrees and,
hence, we need to be careful which changes occur after deleting a high edge.
On the other side, the relative distance of nodes which are ``far away'' from $v$ in $G$
before any edge deletion cannot increase too much.
Thus, we simply estimate their new distances in the third step.
More precisely, in the third step we compute a $(1+\varepsilon)$-approximation of
the eccentricity of $v$ after the deletion of a low edge.

\emph{Step 1: Initialization.}
We build the data structure of Bernstein and Karger~\cite{bernstein2009nearly}, in $mn$ time
which for each triplet $(u,v,e)$ can answer queries of the form $\de(u,v)$ in
$O(1)$ time.

We compute the shortest path tree $T = T_v$ of $v$ in time $O(nm)$ and denote its depth by $d_v$.
By traversing $T$ bottom-up, we compute the height $d_e$ of the subtree $T_e$ 
for each tree-edge $e$; this takes time $O(n)$.
We further construct the sets $T_<$ and $T_>$ of high
and low tree-edges, respectively.

Fix $\varepsilon>0$. We construct a set $S_v \subset V$ as follows:
First add $v$ to $S_v$. Then add each $u \in V$ which has the following
two properties: (1) $u$ is at level $i \varepsilon F$ for some
integer\footnote{For readability we leave out the floors, however, we are considering
		the integer levels $\lfloor i \varepsilon F \rfloor$.}
$i > 0$ and
(2) there exists a node $u'$ in the subtree of $u$ in $T_v$, such that $u'$ has distance
$\varepsilon F /2$ in $T_v$ from $u$. Note that we can add the root,
but every other node we add can be charged to the $\varepsilon F/2$ parent nodes that come before it.
Thus, we can have at most $1+\frac{2n}{\varepsilon F}$ nodes in $S_v$.
Note that for every $z \in V$ there exists a $y \in S_v$, s.t.\ $y$ is an
ancestor of $z$ in $T_v$ and there
exists a path from $y$ to $z$ in $T_v$ of length at most $\varepsilon F$.

Using a second bottom-up traversal of $T$, for each tree-edge $e \in T$,
we compute the set $S_e = T_e \cap S_v$,
i.e., the intersection of the vertices in $T_e$ and those in $S_v$.
This can be done in $O(n)$ time by, instead of storing $S_e$ explicitly for each edge $e = (w,w')$,
storing a reference to the set containing the closest children of $w'$ which are in $S_v$;
then $S_e$ can be constructed in $O(|S_e|)$ time by recursively following the references.

For each non-tree-edge $e$, we store $d_v$ as the value for the eccentricity of $v$
when $e$ is deleted.

\emph{Step 2: Handling high edges.}
For each level $j = 1, \dots, F-1$, we proceed as follows.
We consider each tree-edge $e = (w,w') \in T_<$ with $d(v,w) < d(v,w') = j$, there are at most $n$ of these.
We build a graph $G_e$ containing all nodes of $T_e$ together with a additional directed path $P$ of
length $d_e + 4$ with startpoint $r$. The nodes in $P$ are new vertices added to $G_e$. Each edge on $P$ has weight $1$, except the
single edge incident to $r$, which has weight $\dG(v,w') - 1$.
Additionally to the path, the graph contains as edges: (1) all edges from $E$, which have both endpoints in $T_e$,
and (2) for each $e = (z,z')$ which has its startpoint $z \not\in T_e$ and its endpoint
$z' \in T_e$, an edge $(z'', z')$, where $z''$ is the node
on $P$ with distance $\dG(v,z)$ from $r$.

Observe that $G_e$ has the property that all distances after the deletion of $e$ are
maintained \emph{exactly}: By construction, the shortest path
from $r$ to $u \in T_e$ in $G_e$ has exactly length $\de(v,u)$
(we prove this formally in Lemma~\ref{lem:HighEdges}).

After building $G_e$, we compute its depth starting from node $r$ and store this
value for edge $e$.

\emph{Step 3: Handling low edges.}
For each tree-edge $e = (w,w') \in T_>$, we do the following:
Let $S = S_e \cup \{w'\}$. As answer for
a deleted edge $e$, we store $\max\{ d_v, (1+\varepsilon) \max_{y \in S} \de(v,y) \}$.
To determine $\de(v,y)$, we use the data structure of~\cite{bernstein2009nearly}.

\emph{Extension to $(1+\varepsilon)$-approximate Diameter and Radius.}
\label{Sec:ExentensionsDiameterRadius}
We repeat the previously described procedure for all $v \in V$
(but we build the data structure for Bernstein and Karger only once).
To compute the diameter, we keep track of the maximum value we encounter
for each deleted edge $e$.
To compute the radius, we keep track of the minimum value we encounter
for each deleted edge $e$.

\paragraph*{Correctness.}
Observe that it is enough to show correctness for a fixed $v \in V$.
We first prove the correctness of the algorithm after removing high edges.
\begin{lemma}
\label{lem:HighEdges}
	After the deletion of a high edge $e \in T_<$, we compute the eccentricity of $v$ exactly.
\end{lemma}
\begin{proof}
	Consider any vertex $u \in T_e$ and consider the shortest path $p$ from $v$ to $u$
	in $G \setminus e$.

	We can assume that $p$ has exactly one edge $(z,z')$, s.t.\ $z \not\in T_e$
	and $z' \in T_e$: Assume that there is a path $p'$
	with two edges $(x,x')$ and $(y,y')$, s.t.\ $x,y \not\in T_e$, $x',y' \in T_e$
	and $y$ appears later on $p'$ than $x$.
	Since $y \not\in T_e$, there exists a path from $v$ to $y$ that does
	not use any vertex from $T_e$ and that is of the same length as the subpath of
	$p'$ from $v$ to $y$ in $T$, because $T$ is a shortest-path tree with root $v$.
	Hence, we can choose a path to $y$ without entering $T_e$.

	Let $z, z'$ be as before.
	Then $\de(v,u) = \de(v,z) + 1 + \dG(z',u)$.
	By construction of $G_e$, there exists a vertex on $P$ in $G_e$
	with distance $\de(v,z)$ from $r$ and which has an edge to $z'$.
	All paths which only traverse vertices from $T_e$ are unaffected by the deletion of $e$.
	Hence, in $G_e$ there exists a path of length $\de(u,v)$.

	Also, there is no shorter path in $G_e$ from $r$ to $u$, because this would imply a shorter path
	in $G \setminus e$ by construction.
\end{proof}

Next we prove the correctness of the algorithm after the removal of low edges.
Consider a tree-edge $e = (w,w') \in T_>$ with
$F \geq d(v,w') > d(v,w)$.
Let $S = S_e \cup \{w'\}$.
\begin{lemma}
\label{lem:CloseVertexFromS}
	For each node $z \in T_e$, there exists a vertex $y \in S$
	s.t.\ $\de(y,z) \leq \varepsilon F$.
\end{lemma}
\begin{proof}
	Since $w' \in S$, the claim is true for all nodes $z \in T_e$ with
	$d(w',z) \leq \varepsilon F$. By construction of $S_v$, any (directed) tree path
	of length $\varepsilon F$ contains a node of $S_v$.
	For any node $z \in T_e$ with $d(w',z) > \varepsilon F$, there exists an ancestor
	$u$ of $z$ in $T_e$ with $d(u,z) \leq \varepsilon F$.
	The path from $u$ to $z$ is a directed tree path and, thus, must contain a node in $S_v$.
	Thus, for each node in $T_e$
	there is a path of length at most $\varepsilon F$ from some node in $S_v$.
\end{proof}

\begin{lemma}
\label{lem:DetourOffset}
	Consider two vertices $y,z \in T_e$ and assume there exists a path from $y$ to $z$ in $G \setminus e$.
	Then $\de(y,z) \leq X$ implies $\de(v,z) \leq \de(v,y) + X$.
\end{lemma}
\begin{proof}
	Concatenate the shortest paths from $v$ to $y$ in $G \setminus e$
	and from $y$ to $z$ in $G \setminus E$, which both avoid $e$.
	This path cannot be shorter than the shortest path from $v$ to $z$ in $G \setminus e$.
\end{proof}

We define the maximum height achieved by the vertices of $T_e$ in $G \setminus e$ by
\begin{align*}
	n(v,e) = \max_{z \in T_e} \de(v,z).
\end{align*}
Notice that the eccentricity of $v$ in $G \setminus e$ is given by $\max\{ d_v, n(v,e) \}$.
Hence, by giving a $(1 + \varepsilon)$-approximation of $n(v,e)$, we obtain a
$(1+\varepsilon)$-approximation for the eccentricity of $v$ in $G \setminus e$.
In the remainder of this subsection, we show this guarantee
on the approximation ratio of $n(v,e)$.

\begin{lemma}
	$n(v,e) \leq (1+\varepsilon) \max_{y \in S} \de(v,y)$.
\end{lemma}
\begin{proof}
	Let $z'$ be any vertex in $T_e$ such that $\de(v,z') = n(v,e)$.
	By Lemma~\ref{lem:CloseVertexFromS} there exists a vertex $y' \in S$ with
	$\de(y',z') \leq \varepsilon F$.
	Then by Lemma~\ref{lem:DetourOffset},
	\begin{align*}
		n(v,e) &= \de(z',v) \\
				&\leq \de(v,y') + \varepsilon F \\
				&\leq \max_{y \in S} \de(v,y) + \varepsilon \dG(v,w') \\
				&\leq (1 + \varepsilon) \max_{y \in S} \de(v,y),
	\end{align*}
	where in the second last step we used $F \leq \dG(v,w')$ and in the last step
	we used that $w' \in S$.
\end{proof}

\begin{lemma}
	$(1+\varepsilon) \max_{y \in S} \de(v,y) \leq (1+\varepsilon) n(v,e)$.
\end{lemma}
\begin{proof}
	This follows from the definition of $n(v,e)$, since $S$ is a subset
	of the vertices in $T_e$.
\end{proof}

\paragraph*{Running Time Analysis.}
\label{Sec:RunningTime}
Let us first consider the time spent on step 1, preprocessing.
We build the data structure of Bernstein and Karger~\cite{bernstein2009nearly}
in time $\tO(mn)$.
For each node $v$, computing the shortest path tree of $v$ takes time $\tO(m)$
and we spend time $O(n)$ computing the heights of the subtrees of $T$ and
computing the sets $T_<$, $T_>$, $S_v$. The sets $S_e$ can as well be computed
in $O(n)$ time by storing them only implicitly.

Now let us consider the time spent on step 2, the high edges.
For the high edges $e$ at level $j \leq F$, observe that the trees $T_e$ are mutually disjoint.
Hence, for a fixed level $j$, in time $\tO(m)$ we can compute the depths of \emph{all} graphs $G_e$ with $e$ at level $j$.
Since we have to do this for each level less than $F$, the total time for this step is $\tO(Fm)$.

Finally, let us consider the time spent on step 3, the low edges. For all low edges at level $j > F$, we query all nodes of $S_v$ with height more than $j$.
These are $O(\frac{n}{\varepsilon F})$ many such nodes.
Thus, the total time we spend for all edges in $T$ is $O(d_v \cdot \frac{n}{\varepsilon F})$.

To compute the diameter, we have to execute the above steps once for each $v \in V$,
but we only need to build the data structure of Bernstein and Karger once.
Hence, the total time is $\tO(mn + Fmn + n d_v \cdot \frac{n}{\varepsilon F})$.
Denote the diameter of $G$ by $D$. Then setting $F = \sqrt{ \frac{D n}{\varepsilon m} }$
yields a total running time of $\tO(n^{1.5} \sqrt{D m / \varepsilon}+mn)$.
Since $D \leq n$, this is $\tO( n^2 \sqrt{m / \varepsilon})$.

\bibliography{bibliography}{}

\appendix
\section{Appendix}

\subsection{Definitions of the Problems}
We give definitions of the problems we consider in this paper in Table~\ref{Tbl:Problems}.

\begin{table}[htb!]
\begin{tabular}{|c|c|c|}
\hline
\multicolumn{3}{|c|}{\textbf{Problem}} \\
\hline
Maintain & Update & Query \\
\hline
\multicolumn{3}{|c|}{Reachability} \\
\hline
Directed graph & Edge insertions/deletions & Given two vertices $u, v$, \\
& & can $v$ be reached from $u$? \\
\hline
\multicolumn{3}{|c|}{\#SSR} \\
\hline
Directed graph and a fixed & Edge insertions/deletions & How many vertices can be \\
source vertex $s$. & & reached from $s$? \\
\hline
\multicolumn{3}{|c|}{Strong Connectivity (SC)} \\
\hline
Directed graph & Edge insertions/deletions & Is the graph strongly connected? \\
\hline
\multicolumn{3}{|c|}{2 Strong Components (SC2)} \\
\hline
Directed graph & Edge insertions/deletions & Are there more than $2$ SCCs? \\
\hline
\multicolumn{3}{|c|}{2 vs $k$ Strong Components (AppxSCC)} \\
\hline
Directed graph & Edge insertions/deletions & Is the number of SCCs $2$ \\
&& or more than $k$? \\
\hline
\multicolumn{3}{|c|}{Maximum SCC Size (MaxSCC)} \\
\hline
Directed graph & Edge insertions/deletions & What is the size of the \\
&& largest SCC? \\
\hline
\multicolumn{3}{|c|}{Subgraph Connectivity} \\
\hline
Fixed undirected graph, & Turn on/off vertex & Given two vertices $u,v$, \\
with some vertices on and&& are $u$ and $v$ connected by a path \\
some off. && only traversing vertices that are on? \\
\hline
\multicolumn{3}{|c|}{$\alpha$-approximate Shortest Paths} \\
\hline
Directed or undirected & Edge insertions/deletions & Given two vertices $u,v$,  \\
(possibly weighted) graph & & return an $\alpha$-approximation of the \\
&& length of the shortest path from $u$ to $v$. \\
\hline
\multicolumn{3}{|c|}{$\alpha$-approximate Eccentricity} \\
\hline
Undirected graph & Edge insertions/deletions & Given a vertex $u$,  \\
& & return an $\alpha$-approximation of the \\
&& eccentricity of $v$. \\
\hline
\multicolumn{3}{|c|}{$\alpha$-approximate Radius} \\
\hline
Undirected graph & Edge insertions/deletions & Return an $\alpha$-approximation of the \\
&& radius of the graph. \\
\hline
\multicolumn{3}{|c|}{$\alpha$-approximate Diameter} \\
\hline
Undirected graph & Edge insertions/deletions & Return an $\alpha$-approximation of the \\
&& diameter of the graph. \\
\hline
\multicolumn{3}{|c|}{Bipartite Perfect Matching (BPMatch)} \\
\hline
Undirected bipartite graph & Edge insertions/deletions & Does the graph have a\\
&& perfect matching? \\
\hline
\multicolumn{3}{|c|}{Bipartite Maximum Weight Matching (BWMatch)} \\
\hline
Undirected bipartite graph & Edge insertions/deletions & Return the weight of the \\
with integer edge weights && maximum weight perfect matching. \\
\hline
\end{tabular}
\caption{The problems we consider in this paper.}
\label{Tbl:Problems}
\end{table}

\subsection{A Note on Terminology}
\label{Sec:Terminology}
The terminology used in the literature for dynamic data structures
in the spirit of Section~\ref{Sec:Introduction} is not consistent.
The phrases which are used contain ``fault-tolerant algorithms'', ``algorithms with sensitivity'',
``algorithms for emergency planning'' and ``algorithms for failure prone graphs''.

In the community of spanners and computational graph theory, it is common to speak about
``fault-tolerant subgraphs''. In this area, this term is used consistently.

In the dynamic graph algorithms community, multiple phrases have been used to describe
algorithms for the model proposed in Section~\ref{Sec:Introduction}.
First, the field was introduced by~\cite{patrascu2007planning} as algorithms
for ``emergency planning''. Later, the terminologies ``sensitivity'' and ``failure prone graphs''
were used (e.g.,~\cite{chechik2012sensitivity,chechik2017approximate,duan2010connectivity,duan2017connectivity}).
When the number of failures in the graph was fixed (e.g. 1 or 2), then often this was stated
explicitly (without further mentioning sensitivity or failure prone graphs).
However, it appears that in the dynamic graph algorithms community the phrase ``sensitivity''
is the most widely used one.

\subsection{A Note on Adversaries}
\label{Sec:Adversaries}

Some of the conditional lower bounds we obtain are for \emph{randomized} algorithms.
Previous literature~\cite{abboud2014popular,henzinger2015unifying,kopelowitz2016higher}
also gave conditional lower bounds for randomized dynamic algorithms;
however, it was not discussed under which kind of adversary the obtained lower bounds hold.
This depends on the conjecture from which the lower bound was obtained.
We observe that in reductions from the static triangle problem, the only
randomness is over the input distribution of the static problem. Hence, for lower bounds
from the triangle conjecture, we can assume an oblivious adversary.
Furthermore, we assume the OMv conjecture in its strongest possible form, i.e.
for oblivious adversaries.
(In~\cite{henzinger2015unifying} the authors did not explicitly state which kind of adversary they
assume for their conjecture.)
Thus, all conditional lower bounds we obtain for randomized algorithms
hold for oblivious adversaries.
Note that a lower bound which holds for oblivious adversaries must always hold for
non-oblivious ones.

We would like to point out another subtlety of our lower bounds:
In reductions from the triangle detection conjecture, the running time
of the algorithm is assumed to be a random variable, but the algorithm must always
answer correctly. However, in reductions from OMv the running time of the algorithm
is determinstic, but the probability of obtaining a correct answer must be at least 2/3.

\subsection{Our Lower Bounds}
In Table~\ref{Tbl:LowerBounds} we summarize summarize our lower bounds for
sensitivity data structures. Table~\ref{Tbl:StaticLowerBounds} states our lower bounds
for static oracle data structures.

\begin{table}[htb]
\begin{tabular}{|c|c|c|ccc|c|c|c|}
	\hline
	\multirow{2}{*}{\textbf{Problem}}&
		\multirow{2}{*}{\textbf{Inc/Dec}}&
		\textbf{Query}&
		\multicolumn{3}{c|}{\textbf{Lower Bounds}} & 
		\multirow{2}{*}{\textbf{Sens.}} & \multirow{2}{*}{\textbf{Conj.}} &
		\multirow{2}{*}{\textbf{Cite}} \\
		& & \textbf{Type} &
		$p(m,n)$ & $u(m,n)$ & $q(m,n)$ & 
		& & \\
	\hline
	Reachability & Inc & $ss$ &
		$\mathbf{n^{3-\varepsilon}}$ & $\mathbf{n^{2-\varepsilon}}$ & $\mathbf{n^{2-\varepsilon}}$ & 
		$1$ & BMM & 
		Theorem~\ref{Thm:TriangleLBs} \\ 
	Reach., SC, BPMatch & & $st$ &
		$\mathbf{n^{3-\varepsilon}}$ & $\mathbf{n^{2-\varepsilon}}$ & $\mathbf{n^{2-\varepsilon}}$ & 
		$2$ & & 
		\cite{abboud2014popular} \\ 
	\hline
	$(3/2-\varepsilon)$-sh. paths & Inc & $ss$ &
		$\mathbf{n^{3-\varepsilon}}$ & $\mathbf{n^{2-\varepsilon}}$ & $\mathbf{n^{2-\varepsilon}}$ & 
		$1$ & BMM & 
		Theorem~\ref{Thm:TriangleLBs} \\ 
	(und. unw.) & & & & & & & & \\ 
	\cline{3-9}
	$(7/5-\varepsilon)$-sh. paths & & $st$ &
		$\mathbf{n^{3-\varepsilon}}$ & $\mathbf{n^{2-\varepsilon}}$ & $\mathbf{n^{2-\varepsilon}}$ & 
		$2$ & BMM & 
		Theorem~\ref{Thm:TriangleLBs} \\ 
	(und. unw.) & & & & & & & & \\ 
	\hline
	\#SSR & Inc & $ss$ &
		$\mathbf{n^t}$ & $\mathbf{m^{1-\varepsilon}}$ & $\mathbf{m^{1-\varepsilon}}$ & 
		$K(\varepsilon,t)$ & SETH & 
		Lemma~\ref{Lem:CountSSR} \\ 
	\hline
	reachability, & Inc & $ST$ &
		$n^t$ & $n^{1-\varepsilon}$ & $n^{1-\varepsilon}$ & 
		$K(\varepsilon,t)$ & SETH & 
		Lemmas~\ref{Lem:STReach} and \ref{Lem:ApproxDiameter} \\ 
	\cline{4-9}
	$\left(4/3-\varepsilon\right)$-diameter & & - &
		$\poly(\mathbf{n})$ & $\mathbf{n^{2-\varepsilon}}$ & $\mathbf{n^{2-\varepsilon}}$ & 
		$\omega(\log n)$ & SETH & 
		\cite{abboud2014popular} \\ 
		for sparse graphs & & - &
		 &  &  & 
		 &  & 
		\\ 
	\hline
	SC2, AppxSCC,  & Inc & - &
		$\poly(\mathbf{n})$ & $\mathbf{m^{1-\varepsilon}}$ & $\mathbf{m^{1-\varepsilon}}$ & 
		$\omega(\log n)$ & SETH & 
		\cite{abboud2014popular} \\ 
		and MaxSCC & & - &
		 &  & & 
		 &   & 
		 \\ 
	\hline
	subgraph conn. & Inc & $st$ &
		$\poly(\mathbf{n})$ & $\mathbf{\poly(d)}$ & $\mathbf{d^{1-\varepsilon}}$ & 
		$d$ & OMv & 
		Theorem~\ref{Thm:SensitivityD} \\
	($\implies$ reachability, & & & 
		& & & 
		& & 
		\\
	\cline{4-9}
	BPMatch, SC) & & &
		$n^{2-\varepsilon}$ & $n^{1-\varepsilon}$ & $d^{1-\varepsilon}$ & 
		$d$ & 3SUM & 
		\cite{kopelowitz2016higher} \\ 
	\hline
	$(2-\varepsilon)$-sh. paths & Inc & $ss$ &
		$\poly(n)$ & $\poly(d)$ & $d^{1-\varepsilon}$ & 
		$d$ & OMv & 
		Theorem~\ref{Thm:SensitivityD} \\ 
	$(5/3-\varepsilon)$-sh. paths & Inc & $st$ &
		$\poly(n)$ & $\poly(d)$ & $d^{1-\varepsilon}$ & 
		$d$ & OMv & 
		Theorem~\ref{Thm:SensitivityD} \\ 
	($\implies$ BWMatch) & & & 
		& & & 
		& & 
		\\ 
	\hline
	\hline
	diameter & & & & & & & & \\ %
	$(4/3 - \varepsilon)$, und. unw. & Dec & - &
		$\mathbf{n^{3-\varepsilon}}$ & $\mathbf{n^{2-\varepsilon}}$ & $\mathbf{n^{2-\varepsilon}}$ & 
		$1$ & BMM & 
		Theorem~\ref{Thm:TriangleLBs} \\ 
	dir. $\&$ und. w. & Dec & - &
		$\mathbf{n^{3-\varepsilon}}$ & $\mathbf{n^{2-\varepsilon}}$ & $\mathbf{n^{2-\varepsilon}}$ & 
		$1$ & APSP & 
		Section~\ref{subsec:APSPdecDiam} \\
	\hline
	$(4/3 - \varepsilon)$-ecc. & Dec & - &
		$\mathbf{n^{3-\varepsilon}}$ & $\mathbf{n^{2-\varepsilon}}$ & $\mathbf{n^{2-\varepsilon}}$ & 
		$1$ & BMM & 
		Theorem~\ref{Thm:TriangleLBs} \\ 
	\hline	
	weighted-ecc. & Dec & - &
		$\mathbf{n^{3-\varepsilon}}$ & $\mathbf{n^{2-\varepsilon}}$ & $\mathbf{n^{2-\varepsilon}}$ & 
		$1$ & APSP & 
		Lemma~\ref{lem:APSPdecDiam} \\ 
	\hline
	shortest paths & & & & & & & & \\ %
	dir. w. & Dec & $st$ &
		$\mathbf{n^{3-\varepsilon}}$ & $\mathbf{n^{2-\varepsilon}}$ & $\mathbf{n^{2-\varepsilon}}$ & 
		$1$ & APSP & 
		\cite{williams2010subcubic} \\
	und. w. & Dec & $st$  &
		$\mathbf{n^{3-\varepsilon}}$ & $\mathbf{n^{2-\varepsilon}}$ & $\mathbf{n^{2-\varepsilon}}$  & 
		$2$ & APSP & 
		Section~\ref{subsubsec:APSPshortpathsDec2} \\
	\hline
	reachability & Dec & $st$ &
		$\mathbf{n^{3-\varepsilon}}$ & $\mathbf{n^{2-\varepsilon}}$ & $\mathbf{n^{2-\varepsilon}}$ & 
		$\Omega(\log n)$ & BMM & 
		\cite{abboud2014popular} \\ 
	($\implies$ SC, BPMatch) &  &  &
		 &  &  & 
		& BMM & 
		\cite{abboud2014popular} \\ 
	\hline
	shortest paths & Dec & $st$ &
		$\mathbf{n^{3-\varepsilon}}$ & $\mathbf{n^{2-\varepsilon}}$ & $\mathbf{n^{2-\varepsilon}}$ & 
		$\Omega(\log n)$ & BMM & 
		Theorem~\ref{Thm:TriangleLBs} \\ 
	(undir. unw.) & & & & & & & &  \\
	\hline
	subgraph conn. & Dec & $st$ &
		$\poly(\mathbf{n})$ & $\mathbf{\poly(d)}$ & $\mathbf{d^{1-\varepsilon}}$ & 
		$d$ & OMv & 
		Theorem~\ref{Thm:SensitivityD} \\
	($\implies$ reachability, & & & 
		& & & 
		& & 
		\\
	\cline{4-9}
	BPMatch, SC) & & &
		$n^{2-\varepsilon}$ & $n^{1-\varepsilon}$ & $d^{1-\varepsilon}$ & 
		$d$ & 3SUM & 
		\cite{kopelowitz2016higher} \\ 
	\hline
	$(2-\varepsilon)$-sh. paths & Dec & $ss$ &
		$\poly(n)$ & $\poly(d)$ & $d^{1-\varepsilon}$ & 
		$d$ & OMv & 
		Theorem~\ref{Thm:SensitivityD} \\ 
	$(5/3-\varepsilon)$-sh. paths & Dec & $st$ &
		$\poly(n)$ & $\poly(d)$ & $d^{1-\varepsilon}$ & 
		$d$ & OMv & 
		Theorem~\ref{Thm:SensitivityD} \\ 
	($\implies$ BWMatch) & & & 
		& & & 
		& & 
		\\ 
	\hline
\end{tabular}
\caption{The conditional lower bounds we obtained for non-zero sensitivity.
	Problems for which there exists a tight upper bound are marked bold.
	Regarding the sensitivities, the lower bounds hold for any data structure that supports
	\emph{at least} the sensitivity given in the table;
	$d$ is a parameter that can be picked arbitrarily, and
	$K(\varepsilon,t)$ is a constant depending on properties of SAT and
	the allowed preprocessing time (see Section~\ref{Sec:SETHSensitivity}).
	Lower bounds for constant sensitivities hold in particular for any dynamic algorithm which
	allows for \emph{any} larger fixed constant sensitivity or sensitivity $\omega(1)$.
	For the query type we use the following abbreviations: $st$ -- fixed source and sink,
	$ss$ -- single source, $ap$ -- all pairs, $ST$ -- a fixed set of sources and a fixed set of sinks.
	The rest of the abbreviations are as follows: sh. paths means shortest paths, conn. means connectivity, SC means strongly connected components, SC2 means whether the number of strongly connected components is more than 2, Reach. means Reachability, BPMatch is bipartite matching, BWMatch is bipartite maximum weight matching, ecc. is eccentricity, dir. means directed, und. means undirected, w. means weighted, unw. means unweighted, Conj. means Conjecture.}
\label{Tbl:LowerBounds}
\end{table}




\begin{table}[htb]
	\begin{tabular}{|c|c|c|ccc|c|c|}
		\hline
		\multirow{2}{*}{\textbf{Problem}}&
		\textbf{Inc/Dec/}&
		\textbf{Query}&
		\multicolumn{3}{c|}{\textbf{Lower Bounds}} & 
		 \multirow{2}{*}{\textbf{Conj.}} &
		\multirow{2}{*}{\textbf{Cite}} \\
		&\textbf{Static} & \textbf{Type} &
		$p(m,n)$ & $u(m,n)$ & $q(m,n)$ & 
		& \\

	\hline
	Reach. & static & $ap$ &
	$\mathbf{n^{3-\varepsilon}}$ & - & $\mathbf{n^{2-\varepsilon}}$ & 
	BMM & 
	Theorem~\ref{Thm:TriangleLBs} \\ 
	\hline
	$(5/3 - \varepsilon)$-sh. paths & static & $ap$ &
	$\mathbf{n^{3-\varepsilon}}$ & - & $\mathbf{n^{2-\varepsilon}}$ & 
	BMM & 
	Theorem~\ref{Thm:TriangleLBs} \\ 
	\hline
	Repl. paths & static & $st$ &
	$\mathbf{n^{3-\varepsilon}}$ & - & $\mathbf{n^{2-\varepsilon}}$ & 
	APSP & 
	Section~\ref{subsec:APSPproofsApendix} \\ 
	(1 edge fault) & & & 
	& & & 
	&   \\
	(dir. w.) & & & 
	& & & 
	&  \\
	\hline
\end{tabular}
\caption{The conditional lower bounds we obtained for static oracle data structures,
	i.e., data structures with zero sensitivity.
	Problems for which there exists a tight upper bound are marked bold.
	The query type ``ap'' denotes all pairs queries, Repl. means replacement, the rest of the abbreviations are as in Table 1.}
\label{Tbl:StaticLowerBounds}
\end{table}


\subsection{Existing Sensitivity Data Structures}
\label{Sec:UpperBounds}
In Table~\ref{Tbl:UpperBounds} we summarize existing sensitivity data structures.

In the table, we also list algorithms for ``fault-tolerant subgraphs'' although they are not algorithms
for the sensitivity setting in the classical sense.
However, the fault-tolerant subgraphs are often much smaller than the input graphs
and by traversing the fault-tolerant subgraph during queries, one can obtain
better query times than by running the static algorithm on the original graph.
Unfortunately, the construction time of these subgraphs is often very expensive, though still polynomial;
the goal of these papers is to optimize the trade-offs between the size of the
subgraphs and the approximation ratios achieved for the specific problem.

It is striking that (to the best of our knowledge) most of the existing algorithmic work was obtained
for the case of \emph{decremental} algorithms with a limited number of failures.
While this is natural for the construction of fault-tolerant subgraphs,
this is somewhat surprising from an algorithmic point of view.
Our lower bounds might give an explanation of this phenomenon as they
indicate that for many problems there is a natural bottleneck
when it comes to the insertion of edges.

\begin{landscape}


\begin{table}[htb]
\begin{tabular}{|c|c|cccc|c|c|c|}
	\hline
	\multirow{2}{*}{\textbf{Problem}} &
		\multirow{2}{*}{ {\small \textbf{Approx.}}} & 
		\multicolumn{4}{c|}{\textbf{Upper Bounds}} & 
		\textbf{Sensi-} &
		\multirow{2}{*}{\textbf{Ref.}} & \multirow{2}{*}{\textbf{Remark}} \\
	&
	& Space & $p(m,n)$ & $u(m,n)$ & $q(m,n)$ & \textbf{tivity} 
		& & \\
	\hline
	Dec ap-Connectivity &
		& $n$ & 
		$\poly(n)$ & $d$ & $1$ & 
		$d$ & 
		\cite{patrascu2007planning} & \\ 
	\hline
	Dec ap-SubgraphConn &
		& $d^{1-2/c}mn^{1/c}$ & 
		$d^{1-2/c}mn^{1/c}$ & $d^{4+2c}$ & $d$ & 
		$d$ & 
		\cite{duan2010connectivity} & {\small Any $c \in \mathbb{N}$ can be picked;} \\ 
	&
		&  & 
		 &  &  & 
	& 
		 & {\small  space simplified.} \\ 
	&
		& $dm$ & 
		$mn$ & $d^3$ & $d$ & 
		$d$ & 
		\cite{duan2017connectivity} & Deterministic. \\ 
	&
		& $m$ & 
		$mn$ & $d^2$ & $d$ & 
		$d$ & 
		\cite{duan2017connectivity} & Randomized. \\ 
	\hline
	Inc ap-SubgraphConn &
		& $n^2$ & 
		$n^3$ & $d^2$ & $d$ & 
		$d$ & 
		\cite{henzinger2016incremental} & \\ 
	\hline
	Fully Dynamic ap-SubgraphConn &
		& $n^2 m$ & 
		$n^3 m$ & $d^4$ & $d^2$ & 
		$d$ & 
		\cite{henzinger2016incremental} & Uses \cite{duan2017connectivity} as a blackbox. \\ 
	\hline
	Dec ss-Reachability &
		& $n$ & 
		$m$ & $1$ & $n$ & 
		$1$ & 
		\cite{baswana2015fault} & Subgraph. \\ 
	(implies SCC, dominator tree) &
		& $n$ & 
		$n$ & \multicolumn{2}{c|}{$1$} & 
		$1$ & 
		\cite{baswana2012single} & Oracle. Planar graph. \\ 
	\cline{2-9}
	&
		& $n$ & 
	$\poly(n)$ & \multicolumn{2}{c|}{$1$} & 
		$2$ & 
		\cite{choudhary2016optimal} & Allows for vertex failures. \\ 
	\cline{2-9}
	 &
		& $2^d n$ & 
		$2^d mn$ & $1$ & $2^d n$ & 
		$d$ & 
		\cite{baswana2016fault} & {\small Subgraph. Reasonable for } \\ 
	&
		 &  & 
		  &  &  & 
		  & 
		  & {\small  $d = o(\log( m/n ))$.} \\ 
	\hline
	Dec ss-SP &
		& & 
		& & & 
		& 
		& \\ 
	undirected unweighted &
		$3$ & $n$ & 
		$m$ & \multicolumn{2}{c|}{$1$} & 
		$1$ & 
		\cite{baswana2013approximate} & Oracle. \\ 
	&
		$1+\varepsilon$ & $n/\varepsilon^3$ & 
		& \multicolumn{2}{c|}{$1$} & 
		$1$ & 
		\cite{baswana2013approximate} & Oracle. \\ 
	\cline{2-9}
	&
		exact & $n^{5/3}$ & 
		$\poly(n)$ & $1$ & $n^{5/3}$ & 
		$2$ & 
		\cite{parter2015dual} & More robust BFS tree. \\
	\cline{2-9}
	directed unweighted &
		exact & $n^2$ & 
		$n^\omega$ & $1$ & $1$ & 
		$1$ & 
		\cite{grandoni2012improved} & Algorithm and oracle. \\ 
		\cline{2-9}
	undirected weighted &
		$1 + \varepsilon$ & $m + n/\varepsilon$ & 
		$mn$ & \multicolumn{2}{c|}{$1/\varepsilon$} & 
		$1$ & 
		\cite{bilo2016compact} & Oracle. \\ 
	&
		$2$ & $m$ & 
		$mn$ & \multicolumn{2}{c|}{$1$} & 
		$1$ & 
		\cite{bilo2016compact} & Oracle. \\ 
	\cline{2-9}
	&
		$2 \mathcal{O} + 1$ & $d n$ & 
		$dm$ & \multicolumn{2}{c|}{$d^2$} & 
		$d$ & 
		\cite{bilo2016multiple} & Oracle. \\ 
	\hline
\end{tabular}
\caption{Upper Bounds. We omit polylog factors in the stated running times and spaces usages.
	We use the following abbreviations:
				``ap'' means ``all pairs'', ``ss'' means ``single source'', ``st'' denotes problems
				with a fixed source and a fixed sink.
				Oracles combine update and query into a single operation.
				For algorithms with an additive approximation
				guarantee, we included the optimal result $\mathcal{O}$; all other
				approximation algorithms achieve multiplicative approximation guarantees. See table \ref{Tbl:UpperBoundsAPSP} for APSP upper bounds. }
\label{Tbl:UpperBounds}
\end{table}

\begin{table}[htb]
	\begin{tabular}{|c|c|cccc|c|c|c|}
		\hline
		\multirow{2}{*}{\textbf{Problem}} &
		\multirow{2}{*}{ {\small \textbf{Approx.}}} & 
		\multicolumn{4}{c|}{\textbf{Upper Bounds}} & 
		\textbf{Sensi-} &
		\multirow{2}{*}{\textbf{Ref.}} & \multirow{2}{*}{\textbf{Remark}} \\
		&
		& Space & $p(m,n)$ & $u(m,n)$ & $q(m,n)$ & \textbf{tivity} 
		& & \\
	\hline
		Dec APSP &
		& & 
		& & & 
		& 
		& \\ 
		unweighted undirected &
		$\mathcal{O} + 2$ & $n^{5/3}$ & 
		$\poly(n)$ & $1$ & $n^{5/3}$ & 
		$1$ & 
		\cite{bilo2015improved} & Additive spanner. \\ 
		&
		$\mathcal{O} + 4$ & $n^{3/2}$ & 
		$\poly(n)$ & $1$ & $n^{3/2}$ & 
		$1$ & 
		\cite{bilo2015improved} & Additive spanner. \\ 
		&
		$\mathcal{O} + 10$ & $n^{7/5}$ & 
		$\poly(n)$ & $1$ & $n^{7/5}$ & 
		$1$ & 
		\cite{bilo2015improved} & Additive spanner. \\ 
		&
		$\mathcal{O} + 14$ & $n^{4/3}$ & 
		$\poly(n)$ & $1$ & $n^{4/3}$ & 
		$1$ & 
		\cite{bilo2015improved} & Additive spanner. \\ 
		&
		$(2k - 1)(1 + \varepsilon)$ & $k n^{1+ (k \varepsilon^4)^{-1}}$ & 
		& \multicolumn{2}{c|}{$k$} & 
		$1$ & 
		\cite{baswana2013approximate} & Oracle. Any $k > 1$ and $\varepsilon > 0$. \\ 
		&
		$3$ & $n$ & 
		$\poly(n)$ & $1$ & $n$ & 
		$1$ & 
		\cite{parter2014fault} & Spanner. \\ 
		\cline{2-9}
		&
		{\tiny $3(d+1)\mathcal{O} + (d+1)\log n$} & $dn$ & 
		$\poly(n)$ & $1$ & $dn$ & 
		$d$ & 
		\cite{parter2014fault} & Spanner. \\ 
		\cline{2-9}
		non-negative weights, undirected &
		$1 + \varepsilon$ & $n / \varepsilon^2$ & 
		$\poly(n)$ & $1$ & $n / \varepsilon^2$ & 
		$1$ & 
		\cite{bilo2014fault} & Spanner. Vertex and \\ 
		 &
		 &  & 
		& & & 
		& 
		 & edge deletions. \\ 
		\cline{2-9}
		weighted, undirected &
		$(8k+2)(d+1)$ & $d k n^{1+1/k}$ & 
		$\poly(n)$ & \multicolumn{2}{c|}{$d$} & 
		$d$ & 
		\cite{chechik2012sensitivity} & Oracle. Any $k \in \mathbb{N}$. \\ 
		&
		$1+\varepsilon$ & {\tiny $d n^2 (\log n/\varepsilon)^d$ } & 
		{\tiny $d n^5 \log(n / \varepsilon)^d$ } & \multicolumn{2}{c|}{$d^5$} & 
		$d$ & 
		\cite{chechik2017approximate} & Oracle. \\ 
		\cline{2-9}
		weighted, directed &
		exact & $ $ & 
		$M n^{2.88}$ & \multicolumn{2}{c|}{$n^{0.7}$} & 
		$1$ & 
		\cite{grandoni2012improved} & Oracle. Weights: $\{-M,\dots,M\}$.\\ 
		&
		& & 
		& & & 
		& 
		&  \small Simplified running times. \\ 
		\cline{9-9}
		&
		exact & $n^2$ & 
		$mn$ & \multicolumn{2}{c|}{$1$} & 
		$1$ & 
		\cite{bernstein2009nearly} & Oracle. \\ 
		\cline{2-9}
		&
		exact & $n^2$ & 
		$\poly(n)$ & \multicolumn{2}{c|}{$1$} & 
		$2$ & 
		\cite{duan2009dual} & Oracle. \\ 
		\cline{2-9}
		&
		$\mathcal{O} + d$ & $d n^{4/3}$ & 
		$\poly(n)$ & $1$ & $n^{4/3}$ & 
		$d$ & 
		\cite{braunschvig2012fault} & Additive spanner.\\ 
		\hline
	\end{tabular}
	\caption{Upper Bounds for APSP. We omit polylog factors in the stated running times and spaces usages.
		For algorithms with an additive approximation
		guarantee, we included the optimal result $\mathcal{O}$; all other
		approximation algorithms achieve multiplicative approximation guarantees.}
	\label{Tbl:UpperBoundsAPSP}
\end{table}

\end{landscape}


\subsection {Triangle Detection Proofs}
\label{subsec:ReachSPAppendix}

We provide full details of the reachability and shortest paths sensitivity results
from Theorem~\ref{Thm:TriangleLBs}.

\paragraph*{Reachability}
Let $G = (V,E)$ be an undirected unweighted graph for Triangle Detection.
We create four copies of $V$ denoted by $V_1, V_2, V_3, V_4$,
and for $i = 1,2,3$, we add edges
between nodes $u_i \in V_i$ and $v_{i+1} \in V_{i+1}$ if $(u,v) \in E$.

For a fixed source $s \in V$ and sink $t \in V$, Abboud and Williams~\cite{abboud2014popular}
give the following reduction:
For each vertex $v \in V$, they insert the edges $(s,v_1)$ and $(v_4,t)$,
and query if there exists a path from $s$ to $t$. They show that there exists a triangle
in $G$ iff one of the queries is answered positively.
We observe that this reduction requires $n$ batch updates of size $2$ and $n$ queries.
Hence, it holds for sensitivity $2$.

Now keep $s$ fixed, but remove the sink $t$, and allow single-source
reachability queries\footnote{Given $v \in V$, a query returns true iff there exists a path from $s$ to $v$.}.
We perform a stage for each $v \in V$, in which we add the single edge $(s,v_1)$
and query if there exists a path from $s$ to $v_4$.
By the same reasoning as before, there exists a triangle in $G$ iff one of the
queries returns true. The reduction requires $n$ updates of size $1$ and $n$ queries.
Thus, it has sensitivity $1$.

Finally, we remove the source node $s$ and ask all-pairs reachability
queries\footnote{Given two nodes $u, v \in V$, a query returns true iff there exists a path from $u$ to $v$.}.
We perform a stage for each $v \in V$, which queries if there exists a path
from $v_1$ to $v_4$. There exists a triangle in $G$ iff one if the queries returns true.
This reduction has sensitivity $0$, i.e., it uses no updates, and $n$ queries.
Hence, we have derived a very simple conditional
lower bound for static reachability oracles.

These reductions prove the first three results of Theorem~\ref{Thm:TriangleLBs}.

\paragraph*{Shortest Paths}
The above reduction for $st$-reachability can be easily altered
to work for $(7/5-\varepsilon)$-approximate
$st$-shortest paths in \emph{undirected} unweighted graphs (for any $\varepsilon>0$) with the same running time lower bounds:
Just observe that the graphs in the reduction are bipartite. Thus, either there is a path from $s$ to $t$ of length $5$ and there is a triangle in the original graph, or the shortest path between $s$ and $t$ has length at least $7$. Thus distinguishing between length $7$ and $5$ solves the triangle problem.

To obtain a lower bound for $(3/2-\varepsilon)$-approximate $ss$-shortest paths,
we take the construction for $ss$-reachability and again observe that the graph is bipartite so that if there is no path of length $4$ between $s$ and a node $v\in V$, then the shortest path between them must have length at least $6$.

With the same bipartiteness observation, we obtain a conditional lower bound for
$(5/3-\varepsilon)$-approximate static $ap$-shortest paths.
In a stage for node $v \in V$, we query the shortest path from $v_1$ to $v_4$.
The query returns $3$ if there exists a triangle in the original graph
containing the vertex $v$ and $\geq 5$ otherwise. Thus, distinguishing between $3$ and $\geq 5$ suffices to solve the problem.

Finally, let us discuss how we obtain a lower bound for incremental shortest paths with sensitivity $1$ in directed graphs.
Vassilevska Williams and Williams~\cite{williams2010subcubic} reduce BMM to the replacement paths problem in directed unweighted graphs: given a directed graph $G$ on $m$ edges and $n$ nodes and two nodes $s$ and $t$, compute for every $e\in E$, the distance between $s$ and $t$ in $G\setminus \{e\}$. \cite{williams2010subcubic} shows that if a combinatorial algorithm can solve the latter problem in $O(mn^{1/2-\varepsilon})$ time for any $\varepsilon>0$ (for any choice of $m$ as a function of $n$), then BMM has a truly subcubic combinatorial algorithm. This showed that the $O(m\sqrt n)$ time algorithm of Roditty and Zwick~\cite{roditty2012replacement} is tight.

Here we observe that the \cite{williams2010subcubic} reduction immediately implies a $1$-sensitivity oracle lower bound, as any $1$-sensitivity oracle for $st$ distances must be able to answer the replacement paths problem by querying the less than $n$ nodes on the shortest $s$-$t$ path: either the preprocessing time is at least $mn^{0.5-o(1)}$ or the query time is at least $m/n^{0.5+o(1)}$. For dense graphs this gives a lower bound of either $n^{2.5-o(1)}$ preprocessing or $n^{1.5-o(1)}$ query time. The lower bound is again tight via Roditty and Zwick's algorithm~\cite{roditty2012replacement} for replacement paths.

\subsection{All Pairs Shortest Paths Proofs}
\label{subsec:APSPproofsApendix}
We prove the statements of Theorem~\ref{Thm:APSPLBs}
in Lemma~\ref{lem:APSPdec2SP} and Lemma~\ref{lem:APSPdecDiam}.

\begin{lemma}
Assuming the APSP conjecture, decremental $st$-shortest paths in undirected weighted graphs
with sensitivity 2 cannot be solved with preprocessing time $O(n^{3-\varepsilon})$,
and update and query times $O(n^{2-\varepsilon})$ for any $\varepsilon > 0$.
\label{lem:APSPdec2SP}
\end{lemma}

\begin{proof}
	We use a reduction similar to the reduction from APSP to RP from~\cite{williams2010subcubic},
	but to deal with the undirected edges we add more weights and we add additional nodes to the graph. 
	As in \cite{williams2010subcubic}, we start by taking an instance of APSP and
	turning it into a tripartite graph for the negative triangle detection problem\footnote{In
		the negative triangle detection problem we are given an edge-weighted graph
		$G = (V,E)$ with possibly negative edge-weights from $\mathbb{Z}$,
		and we must determine if $G$ contains a triangle consisting of vertices $u,v,x$
		such that $w(u,v) + w(v,x)+ w(x,u) < 0$.};
	denote resulting graph $H'$.

	If $H'$ has no negative edge weights, we are done (there are no negative triangles).
	If $H'$ has negative edge weights, let $M=\min\{w(e)|e\in E_{H'}\}$ and add $-M+1$ to all edges
	(thus making all edges have positive weights).
	Now we want to detect if there is a triangle with (positive) weight less than $-3M+3$.
	Denote the new graph by $H$ and denote the three tripartite groups $A$, $B$ and $C$;
	each set $A$, $B$ and $C$ has $n$ nodes. 
	
	\begin{figure}[ht]
		\centering
		\includegraphics[width=0.75\textwidth]{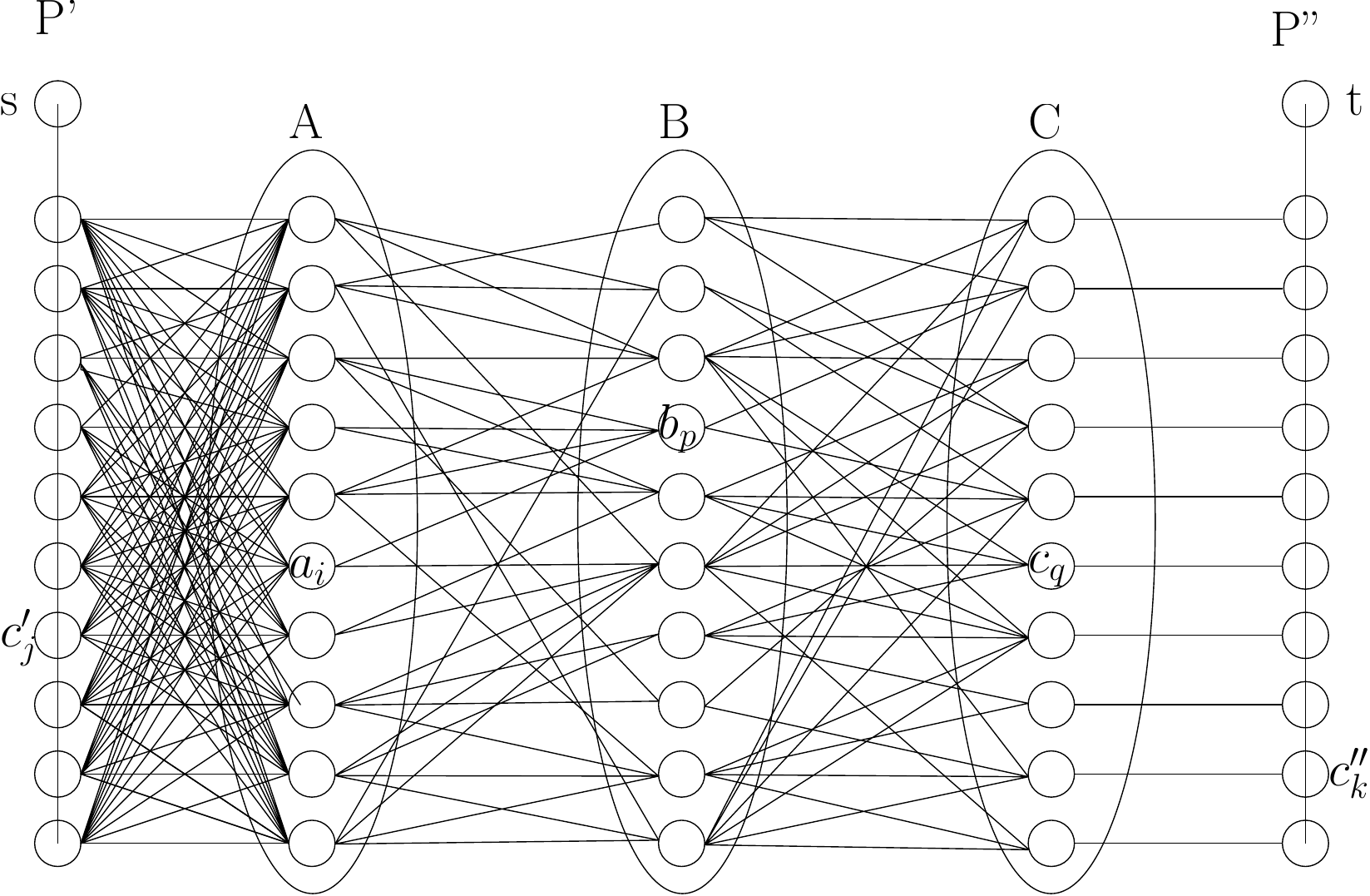}
		\caption{The graph $G$.}
		\label{fig:APSP2del}
	\end{figure}
	
	We construct a graph $G$ in which $n$ shortest paths queries with
	two edge deletions determine if there exist any triangles with weight less than $- 3M+3$ in $H$.
	Let $W = 4 \max \{w(e)| e \in E_{H} \}$ and observe that $W$ is larger than the maximum possible difference
	in the weight of two triangles. We will use this weight to enforce that we must take certain paths.

	We add two vertices $s$ and $t$ to $G$. We add a path $P'$ of length $n$ to $s$,
	where each edge on the path has weight $0$;
	the first node after $s$ on $P'$ is denoted $c_1'$, the next node on $P'$ is denoted $c_2'$,
	and the $i$'th node on $P'$ is denoted $c_i'$.
	Next, we add a path $P''$ of length $n$ to $t$,
	where each edge on the path has weight $0$; the first node after $t$ on $P''$ is denoted $c_n''$,
	second node on $P''$ is denoted $c_{n-1}''$ and the $i^{th}$ node away from $t$ on $P''$ is denoted $c_{n-i+1}''$.
	We add the nodes in $A$, $B$ and $C$ from $H$ to $G$ and keep all edges from $A \times B$ and from $B \times C$,
	however, we delete all edges from $A \times C$.
	We increase the weight of all edges from $A$ to $B$ and of all edges from $B$ to $C$ by $6nW$.
	We add edges between all nodes in $A$ and all nodes on the path $P'$;
	specifically, for all $a_i \in A$ and for all $j \in [1,n]$, we add an edge $(a_i, c_j')$
	of weight $(7n-j)W+w((a_i,c_j))$.
	We further add edges from $C$ to the path $P''$;
	specifically, we add an edge from $c_i \in C$ to $c_i''$ of weight $(6n+i)W$.
	The resulting graph is given in Figure~\ref{fig:APSP2del}.
	
	Note that all edges in the graph $G$ either have weight $0$ or their weight
	is from the range $[6nW, 7nW+W/4]$. All edges of weight $0$ are on the paths $P'$ and $P''$;
	all paths from $s$ to $t$ must contain at least one edge from $P'$ to $A$,
	one from $A$ to $B$, one from $B$ to $C$ and finally one edge from $C$ to $P''$.
	Each of the non-path-edges has weight from the range $[6nW, 7nW+W/4]$ and we must take at least four of them in total.
	If we backtrack (and go from $A$ back to $P'$ or from $B$ back to $A$, etc)
	then we must take at least six non-zero edges;
	hence, it is never optimal to backtrack since $(7nW+W/4)4<6 \dot 6nW$.
	
	We explain which $n$ queries answer the negative triangle question in $H$.
	For each $i = 1,\dots,n$, we delete the edge $(c_i',c_{i+1}')$ from path $P'$
	and the edge $(c_i'',c_{i-1}'')$ from path $P''$, and we query the shortest path from $s$ to $t$.
	Note that with these edges deleted to take only four ``heavy'' edges, one must leave
	from a $c_j'$ where $j\leq i$ and enter a $c_k''$ where $k\geq i$.
	The length from $s$ to $c_j'$ and from $c_k''$ to $t$ is zero. So a shortest path from $s$ to $t$ has length 
	$$(7n-j)W+w(a_p,c_j) + w(a_p,b_q)+6nW + w(b_q,c_k) + 6nW+ (6n+k)W.$$
	Note that because $W$ is large and we want to minimize the length of the shortest path,
	we want to maximize $j$ and minimize $k$.
	Due to the deleted edges, the maximum plausible value of $j$ is $i$ and the
	minimum plausible value of $k$ is $i$. In that case, the path length is
	$$(7n-i)W+w(a_p,c_i) + w(a_p,b_q)+6nW + w(b_q,c_i)+6nW + (6n+i)W .$$
	This simplifies to $25nW+w(a_p,c_i) + w(a_p,b_q) + w(b_q,c_i).$
	If the length of the shortest $st$-path is less than $25nW - 3M+3$, then there exists a negative triangle
	in the graph $H'$ containing $c_i$. Otherwise, there is no such triangle.
\end{proof}

\begin{lemma}
\label{lem:APSPdecDiam}
Assuming the APSP conjecture, decremental diameter in undirected weighted graphs
with sensitivity 1 cannot be solved with preprocessing time $O(n^{3-\varepsilon})$,
and update and query times $O(n^{2-\varepsilon})$ for any $\varepsilon > 0$.
\end{lemma}
\begin{proof} 
	As in \cite{williams2010subcubic} we start by taking an instance of APSP
	and turning it into a weighted tripartite graph for the negative triangle detection problem;
	denote the resulting graph for the negative triangle detection problem $H'$.

	Let $M$ be a positive integer such that all edges in $H'$ have weights between $-M$ and $M$.
	We increase all edge weights in $H$ by $5M$ (thus making all edges have weight at least $4M$).
	We denote this new graph $H$ and call the three tripartite groups $X$, $Y$ and $Z$.
	Note that $X$, $Y$ and $Z$ each have $n$ nodes. 
	Now we want to detect if there is a triangle with weight less than $15M$ in $H$.
	
	We construct the graph $G$ depicted in Figure~\ref{fig:DiamPic} as follows.
	We create sets $V_1$ and $V_4$ containing copies of the nodes in $X$,
	the set $V_2$ containing copies of the nodes from $Y$,
	and the set $V_3$ containing copies of the nodes from $Z$.
	We additionally create two groups of vertices, $A$ and $B$,
	containing copies of $X$.
	Finally, we add two vertices $c$ and $d$ to $G$.
	For convenience, we denote the different copies of $x_i \in X$ as follows: The copy in $V_1$ as $x_i^1$,
	the copy in $V_4$ as $x_i^4$, the copy in $A$ as $x_i^A$, and the copy in $B$ as $x_i^B$. 

	\begin{figure}[htb]
		\begin{centering}
			\includegraphics[width=6.2cm]{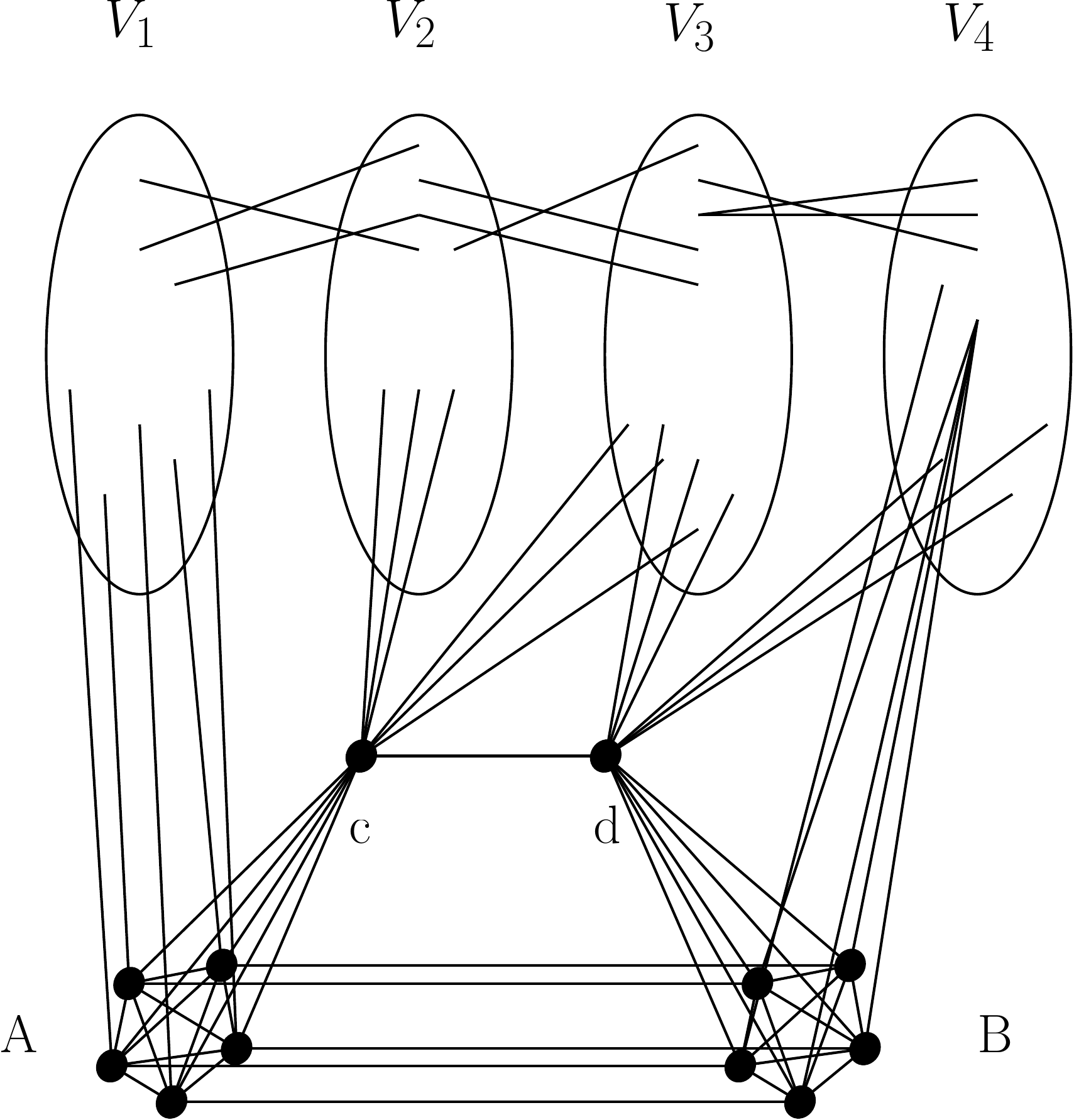}
			\caption{The graph $G'$.}
			\label{fig:DiamPic}
		\end{centering}
	\end{figure}

	We introduce edges between $V_1$ and $V_2$ if the corresponding nodes in $X$ and $Y$ have an edge;
	the edges between $V_2$ and $V_3$ are determined by the edges between $Y$ and $Z$;
	the edges between $V_3$ and $V_4$ are determined by the corresponding edges between nodes of $X$ and $Z$.
	All edges we added to the graph have the same weight as their corresponding copies in $H$.
	Note that there are no edges between $V_1$ and $V_3$, nor between $V_1$ and $V_4$, nor between $V_2$ and $V_4$.
	
	For all $i = 1,\dots,n$, we add edges of weight $4M$ between $x_i^1$ and $x_i^A$, and between
	$x_i^A$ and $x_i^B$. Additionally, for all $i = 1,\dots,n$ and all $j = 1,\dots,n$,
	we add an edge between $x_i^B$ and $x_j^4$ of weight $4M$.
	For all $v \in V_2 \cup V_3 \cup A$, we add an edge of weight $4M$ between $c$ and $v$,
	and for all $v \in V_3 \cup V_4 \cup B$, we add an edge of weight $4M$ between $d$ and $v$.
	We add edges with weight $4M$ to connect all vertices in $A$ into a clique, and we do the same for $B$.
	Finally, we add an edge of weight $4M$ between $c$ and $d$. 
	
	Note that all edges in this graph have weights between $4M$ and $6M$. Further note that in
	$G$ all pairs of nodes have a path of length at most $12M$ between them. 
	
	Our $n$ queries are picked as follows: For all $i = 1,\dots,n$, we delete the edge
	between $x_i^B$ and $x_i^4$ and query the diameter.
	Observe that the only path lengths that could become greater
	than $12M$ are between the nodes in $V_1$ and the node $x_i^4$. However, for all $x_j^1$
	where $j \ne i$, the path $x_j^1 \rightarrow x_j^A \rightarrow x_j^B \rightarrow x_i^4$
	still has all of its edges, and there exists a path of length $12M$ for these vertices. 
	
	Thus the only pair of vertices for the path length might increase is $x_i^1$ to $x_i^4$.
	One must use at most three edges to obatain a path of length at most $15M$
	(since any path with four hops has length at least $16M$ because every edge in $G$
	has weight at least $4M$).
	Hence, any path of length less than $15M$ between $x_i^1$ and $x_i^4$
	must go from $x_i^1$ to $V_2$ to $V_3$ to $x_i^4$.
	Thus, if a path of length less than $15M$ exists between $x_i^1$ and $x_i^4$ after the edge
	deletion between $x_i^B$ and $x_i^4$, then there is a negative triangle in $H'$ containing node $x_i$.
	
	Thus, we can detect if any negative triangle exists in the tripartite graph $H'$
	by asking these $n$ queries, determining for all $x_i \in X$
	if a negative triangle containing $x_i$ exists.
\end{proof}

Note that if for all $i = 1,\dots,n$ we delete edge $(x_i^B,x_i^4)$, then the
eccentricity of $x_i^1$ is less than $15M$ if and only if the diameter of
the graph is less than $15M$. Thus, this serves as lower bound for eccentricity as well. 

\subsection{SETH Proofs}
\label{subsec:SETHAppendix}

\begin{lemma}
\label{Lem:CountSSR}
  Let $\varepsilon > 0$, $t \in \mathbb{N}$.
  SETH implies that there exists no
  algorithm for incremental \#SSR with sensitivity $K(\varepsilon, t)$,
  which has preprocessing time $O(n^t)$,
  update time $u(n)$ and query time $q(n)$, such that 
  $\max\{ u(n), q(n) \}= O(n^{1-\varepsilon})$.
\end{lemma}

\begin{proof}
	Set $\delta = (1-\varepsilon)/t$.
	Assume that the $k$-CNFSAT formula $F$ has $c \cdot \tilde n$ clauses for some
	constant~$c$. Partition the clauses into $K = c/\delta$ groups of size $\delta \tilde n$
	and denote these groups by $G_1, \dots, G_K$.
	Further let $U$ be a subset of the $\tilde n$ variables of $F$ of size $\delta \tilde n$,
	and let $\bar U$ denote the set of all partial assignments to the variables in the $U$.
	
	We construct a graph $G$ as the union of $D_\delta$ and $H_\delta$.
	It consists of $\bar U$ and the set $C$ of clauses.
	We direct the edges in $H_\delta$ from $C$ to $\bar U$, and
	add a directed edge from a node $d \in D_\delta$ to
	$c \in C$ iff $c \in d$. We further add a single node $s$ to the graph.
	The resulting graph has
	$n = O(2^{\delta \tilde n})$ nodes and $O(2^{\delta \tilde n} \tilde n)$ edges.
	
	We proceed in stages with one stage for each partial assignment to the
	variables in $V \setminus U$.
	At a stage for a partial assignment $\phi$ to the variables in $V \setminus U$,
	we proceed as follows:
	For each group $G_i \subset C$, we add an edge from $s$ to the largest
	non-empty subset $d_i$ of $G_i$ which only contains clauses that
	are not satisfied by $\phi$, i.e., to the set
	$d_i = \{ c \in G_i : \phi \not\vDash c\}$; if $d_i$ is empty, then we do not introduce an edge.
	Let $D'$ be the set of nodes $d_i$ that received an edge from $s$ and let $d(s) = |D'|$.
	Note that $d(s) \leq K$, since we introduce at most one edge for each of the $K$ groups.
	Further, let $B$ denote the number of clauses in $C$ reachable from the sets $d_i \in D'$,
	i.e., $B = \sum_{i = 1}^K |d_i|$.
	We query if the number of nodes reachable from $s$ is less than
	$d(s) + B+ 2^{\delta \tilde n}$. If the answer to the query is true, then we return that $F$ is
	satisfiable, otherwise, we proceed to the next partial assignment to the variables in
	$V \setminus U$.
	
	We prove the correctness of the reduction: Assume that $F$ is satisfiable.
	Then there exist partial assignments $\phi$ and $\phi'$ to the variables in
	$V \setminus U$ and $U$, such that $\phi \cdot \phi'$ satisfies $F$.
	Hence, for each subset of clauses $d \subset C$ we have that each clause $c \in d$ is satisfied by $\phi$
	or by $\phi'$. Thus, the node $u \in \bar U$ corresponding to $\phi'$
	cannot be reachable from $s$ and there must be less than $d(s) + B + 2^{\delta \tilde n}$
	nodes reachable from $s$.
	Now assume that at a stage for the partial assignment $\phi$ the result to
	the query is true, i.e., less than $d(s) + B + 2^{\delta \tilde n}$ nodes are reachable from $s$;
	namely $d(s)$ at distance $1$, $B$ at distance $2$, and less than $2^{\delta \tilde n}$ at distance $3$.
	In this case, there must be a node $u \in \bar U$ which is not reachable from
	$s$: In $D_\delta$ there are exactly $d(s)$ nodes reachable and in $C$ there are exactly
	$B$ nodes reachable by construction of the graph and definition of $d(s)$ and $B$.
	This implies that for the partial assignment  $\phi'$ corresponding to $u$,
	each clause $c \in C$ must be satisfied by $\phi$ or $\phi'$.
	Hence, $F$ is satisfiable.

	Note that determining the sets $d_i \in D'$ can be done in time $O(\delta \tilde n^2)$ per group $G_i$
	as for each clause we can check in time $O(\tilde n)$ whether it is satisfied by $\phi$.
	Thus the set $D'$ and the value $B$ can be computed in total time $O(cn)$ and
	the total time for all stages is
	$O(2^{1 - \delta \tilde n} (\tilde n^2 + K u(n) + K q(n))$.
	If both $u(n)$ and $q(n)$ are $O(n^{1-\varepsilon}) = O(2^{\delta \tilde n (1-\varepsilon)})$,
	then SAT can be solved in time $O^*(2^{\tilde n(1-\varepsilon \delta)})$.
\end{proof}

\begin{lemma}
\label{Lem:ApproxDiameter}
  Let $\varepsilon > 0$, $t \in \mathbb{N}$.
  SETH implies that there exists no $(4/3 - \varepsilon)$-approximation
  algorithm for incremental diameter with sensitivity $K(\varepsilon, t)$, which
  has preprocessing time $O(n^t)$,
  update time $u(n)$ and query time $q(n)$, such that 
  $\max\{ u(n), q(n) \}= O(n^{1-\varepsilon})$.
\end{lemma}

\begin{proof}
	Set $\delta = (1-\varepsilon)/t$.
	During the proof we will assume that the $k$-CNFSAT formula $F$
	has $c \cdot \tilde n$ clauses for some constant $c$.
	We partition these clauses into $K = c/\delta$ groups
	of size $\delta \tilde n$ and denote these groups by $G_1, \dots, G_{K}$.
	
	We construct a graph as follows:
	For $U \subset V$ of size $\delta \tilde n$ we create the graph $H_\delta$
	with the set of all partial assignments to the variables in $U$ denoted by $\bar U$ and the set of clauses $C$.
	We also add the graph $D_\delta$ containing all the subsets of the groups $G_i$.
	Furthermore, we introduce four additional nodes $x$, $y$, $z$ and $t$.
	Observe that the graph has $O(2^{\delta \tilde n})$ vertices.
	
	We add an edge from $x$ to each node in $\bar U$.
	We connect $y$ to all nodes in $C$ and to all nodes in $D_\delta$.
	We add further edges between a clause $c$ and a set $g \in D_\delta$ if $c \in g$.
	We also add the following edges to $E$: $\{x,y\}$, $\{y,z\}$ and $\{z,t\}$.
	Hence, we have $O(2^{\delta \tilde n}\tilde n)$ edges in total.
	
	If during the construction of the graph we encounter that a
	clause is satisfied by all partial assignments in $\bar U$,
	then we remove this clause.
	Also, if there exists a partial assignment from $\bar U$
	which satisfies all clauses, we return that the formula is
	satisfiable. Thus, we can assume that each node in $\bar U$
	must have an edge to a node in $C$ and vice versa.
	
	We proceed in stages with one stage for each partial assignment to the
	variables in $V \setminus U$. Denote the partial assignment of the
	current stage by $\phi$.
	For each $G_i$, we add an edge between $t$ and the subset of $G_i$
	that contains all clauses of $G_i$ which are not satisfied
	by $\phi$, i.e., we add an edge between
	$t$ and $\{c \in G_i : \phi \not\vDash c\}$ for all $i = 1,\dots,K$.
	Hence, in each stage we have $O(K) = O(1)$ updates.
	We query the diameter of the resulting graph.
	The diameter is $3$, if the formula $F$ is not satisfiable,
	and it is $4$, otherwise.
	After that we remove the edges that were added in the update.
	
	We prove the correctness of our construction:
	Observe that via $x$ and $y$ all nodes from $\bar U \cup C \cup G$
	have a distance of at most $3$.
	From $z$ we can reach all vertices of $G$ and $C$ via $y$ within two
	steps and all nodes of $\bar U$ within three steps via $y$ and $x$.
	From $t$ we can reach all nodes of $\{x,y,z\} \cup C \cup G$ within
	three steps using the path $t \rightarrow z \rightarrow y \rightarrow v$,
	where $v \in C \cup D_\delta \cup \{x\}$.
	Hence, all nodes in $\{x,y,z,t\} \cup C \cup G$ have a maximum distance of 3.
	From $\bar u \in \bar U$ we can reach $t$ in four steps with the path
	$\bar u \rightarrow x \rightarrow y \rightarrow z \rightarrow t$.
	
	Assume that for $\bar u$ there exists a path $\bar u \to c \to g \to t$,
	then by construction $\bar u \not\vDash c$ and $\phi \not\vDash c$, since
	$c \in g$. Hence, $F$ is not satisfied by $\bar u \cdot \phi$.
	Thus, if the diameter is $3$, then $F$ is not satisfiable.
	On the other hand, if $F$ is not satisfiable, then for each pair of partial assignments
	$\bar u$ and $\phi$, there must be a clause $c$ which both partial assignments do not satisfy.
	Hence, there must be a path of the form $\bar u \to c \to g \to \phi$ for some $g$
	with $c \in g$ and $g$ has an edge to $t$. Thus, if $F$ is not satisfiable,
	then the graph has diameter $3$.

	The sets $d_i$ can be computed as in the previous proof.
\end{proof}

\begin{lemma}
\label{Lem:STReach}
  Let $\varepsilon > 0$, $t \in \mathbb{N}$.
  SETH implies that there exists no
  algorithm for incremental ST-Reach with sensitivity $K(\varepsilon, t)$,
  which has preprocessing time $O(n^t)$,
  update time $u(n)$ and query time $q(n)$, such that 
  $\max\{ u(n), q(n) \}= O(n^{1-\varepsilon})$.
\end{lemma}
\begin{proof}
	We reuse the graph from the proof of Lemma~\ref{Lem:ApproxDiameter}.
	We update it by removing the vertices $x,y,z$.
	We further set $S = \bar U$ and $T = \{t\}$.
	
	We proceed in stages with one stage for each partial assignment to the
	variables in $V \setminus U$. Denote the partial assignment of the
	current stage by $\phi$.
	For each $G_i$, we add an edge between $t$ and the subset of $G_i$
	that contains all clauses of $G_i$ which are not satisfied
	by $\phi$, i.e.\ we add an edge between
	$t$ and $\{c \in G_i : \phi \not\vDash c\}$ for all $i = 1,\dots,K$.
	Hence, in each stage we have $O(K) = O(1)$ updates.
	We query for ST-Reachability. If the answer is true,
	then $F$ is not satisfiable, otherwise, it is.
	
	We prove the correctness of our construction:
	If $F$ is not satisfiable, then for each pair of partial assignments
	$\bar u$ and $\phi$, there must be a clause $c$ which both partial assignments do not satisfy.
	Hence, there must be a path of the form $\bar u \to c \to g \to \phi$ for some $g$
	with $c \in g$ and $g$ has an edge to $t$. Thus, if $F$ is not satisfiable,
	then all nodes in $S$ will be able to reach $t$.
	On the other hand, assume that for $\bar u$ there exists a path $\bar u \to c \to g \to t$,
	then by construction $\bar u \not\vDash c$ and $\phi \not\vDash c$, since
	$c \in g$. Hence, $F$ is not satisfied by $\bar u \cdot \phi$.
	Thus, if all nodes from $S$ can reach $t$, then $F$ is not satisfiable.

	The sets $d_i$ can be computed as in the first proof of the section.
\end{proof}

\subsection{Conditional Lower Bounds For Variable Sensitivity}
\label{Sec:OMvProofs}

In this section we prove conditional lower bounds for algorithms where the sensitivity is
not fixed, but given a parameter~$d$. Before we give our results, we shortly argue why this setting
is relevant.

First, several results were obtained in the setting with sensitivity $d$. Some of these results are
	by Patrascu and Thorup~\cite{patrascu2007planning} for decremental reachability,
	by Duan and Pettie~\cite{duan2010connectivity,duan2017connectivity} and by
	Henzinger and Neumann~\cite{henzinger2016incremental} for subgraph connectivity
	and by Chechik et al.~\cite{chechik2012sensitivity,chechik2017approximate} for decremental all pairs shortest paths.
	Our lower bounds show that the results in \cite{duan2010connectivity,duan2017connectivity} and
	the incremental algorithm in \cite{henzinger2016incremental} are tight.

Second, when $d$ is not fixed and we can prove a meaningful lower bound,
	this will help us understand whether
	updates or queries are more sensitive to changes of the problem instance.

Third, when $d$ is fixed to a constant, the problems might become
	easier in the sense that constant or polylogarithmic update and query times can be achieved.
	For example, for APSP with single edge failures one can achieve query and update
	times $O(1)$ (see~\cite{bernstein2009nearly}); for APSP with two edge failures
	one can achieve query and update times $O(\log n)$ (see~\cite{duan2009dual}).
	In these cases we cannot prove any non-trivial conditional lower bounds for them.
	However, with an additional parameter $d$ we can derive conditional lower bounds
	which are polynomial in the parameter $d$.

Our results are summarized in the following theorem.
\begin{theorem}
\label{Thm:SensitivityD}
	Under the OMv conjecture for any $\varepsilon > 0$, there exists no algorithm with 
	preprocessing time $\poly(n)$, update time $\poly(d)$
	and query time $\Omega(d^{1 - \varepsilon})$ for the following problems: 
	\begin{compactenum}
		\item Decremental/incremental $st$-SubConn in undirected graphs with sensitivity $d$
		\item decremental/incremental $st$-reach in directed graphs with sensitivity $d$,
		\item decremental/incremental BP-Match in undirected bipartite graphs with sensitivity $d$, and
		\item decremental/incremental SC in directed graphs with sensitivity $d$.
		\item $(2-\varepsilon)$-approximate $ss$-shortest paths with sensitivity $d$ in undirected unweighted graphs,
		\item $(5/3-\varepsilon)$-approximate $st$-shortest paths with sensitivity $d$ in undirected unweighted graphs,
		\item BW-Matching with sensitivity $d$.
	\end{compactenum}
\end{theorem}

\paragraph*{Conditional Lower Bounds for Directed Graphs.}
\label{Sec:OMvExistingWork}
We observe that some existing reductions can be used to obtain conditional lower bounds
for sensitivity problems. In this section, we summarize the results that can be obtained this way.

We call a reduction from a dynamic problem $A$ to another dynamic problem $B$
\emph{sensitivity-preserving} if in the reduction a single update
in problem $A$ propagates to problem $B$ as at most one update.
We observe that the reductions in~\cite{abboud2014popular}
from $st$-subgraph-connectivity to $st$-reachability (Lemma~6.1) and from
$st$-reachability to BP-Match (Lemma~6.2) and SC (Lemma~6.4)
are sensitivity-preserving.
Henzinger et al.~\cite{henzinger2015unifying} give conditional lower bounds for
the $st$-subgraph-connectivity problem with sensitivity $d$.
The previous observations about sensitivity preserving reductions imply that we get the
same lower bounds for st-reach, BP-Match and SC with sensitivity $d$.
This implies the first four points of Theorem~\ref{Thm:SensitivityD}.

The construction of the lower bound in~\cite{henzinger2015unifying} required $d = m^\delta$ for some $\delta \in (0, 1/2]$.
This appears somewhat artifical, since in practice one would rather expect situations with much smaller
values for $d$, e.g., $d = O(1)$ or $d = \poly\log(n)$. However, the lower bound is still interesting because
it applies to all algorithms that allow setting $d = m^\delta$.
For example, the sensitive subgraph connectivty algorithms of Duan and Pettie~\cite{duan2010connectivity,duan2017connectivity}
is tight w.r.t.\ to the above lower bound.

\paragraph*{Conditional Lower Bounds for Undirected Graphs.}
In this subsection, we prove the last three points of Theorem~\ref{Thm:SensitivityD}.
We give a reduction from the $\gamma$-uMv-problem, which
was introduced by Henzinger et al.~\cite{henzinger2015unifying}.
The $\gamma$-uMv-problem is as follows: Let $\gamma > 0$.
An algorithm for the $\gamma$-uMv problem is given an $n_1 \times n_2$ binary matrix $M$
with $n_1 = \lfloor n_2^\gamma \rfloor$, that can be preprocessed.
Then two vectors $u$ and $v$ appear and the algorithm must output the result of the
Boolean vector-matrix-vector-product $u^t Mv$.

Henzinger et al.~\cite[Corollary~2.8]{henzinger2015unifying} show that under the OMv conjecture
for all $\varepsilon > 0$, no algorithm exists for the $\gamma$-uMv-problem that has
preprocessing time $\poly(n_1,n_2)$, computation time $O(n_1^{1-\varepsilon} n_2 + n_1 n_2^{1-\varepsilon})$,
and error probability at most $1/3$.
We give a reduction from the $\gamma$-uMv-problem to $(2-\varepsilon)$-approximate
$ss$-shortest-paths with sensitivity $d$ in the following lemma. 
This lemma and the proof of Corollary~3.12 in~\cite{henzinger2015unifying}
imply the result in Theorem~\ref{Thm:SensitivityD} for $ss$-shortest-paths.

\begin{lemma}
\label{Lem:SensitiveSP}
  Let $\delta \in (0, 1/2]$. Given an algorithm $\mathcal{A}$ for
  incremental/decremental $(2 - \varepsilon)$-$ss$-shortest paths with sensitivity $d$,
  one can solve $(\frac{\delta}{1-\delta})$-uMv with
  parameters $n_1$ and $n_2$ by running the preprocessing step of $\mathcal{A}$ on a graph
  with $O(m)$ edges and $O(m^{1-\delta})$ vertices, then making a single batch update of size
  $O(d)$ and $O(m^{1-\delta})$ queries, where $m$ is such that
  $m^{1-\delta} = n_1$ and $d = m^\delta = n_2$.
\end{lemma}
\begin{proof}
  We prove the lower bound for the incremental problem.

  Let $M$ be a $n_1 \times n_2$ binary matrix for $(\frac{\delta}{1-\delta})$-uMv.
  We construct a bipartite graph $G_M$ from the matrix $M$:
  Set $G_M = ((L \cup R), E)$, where $L = \{l_1, \dots, l_{n_1}\}$,
  $R = \{r_1, \dots, r_{n_2}\}$ and the edges are given by
  $E = \{ (l_i, r_j) : M_{ij} = 1 \}$.
  We add an additional vertex $s$ to $G_M$ and attach a path of length $3$
  to $s$, the vertex on the path with distance $3$ from $s$ has edges to all
  vertices in $L$.
  Observe that $G_M$ has $O(n_1 n_2) = O(m)$ edges and
  $n_1 + n_2 = \Theta(m^\delta + m^{1 - \delta}) = \Theta(m^{1 - \delta})$ vertices.

  When the vectors $u$ and $v$ arrive, we add an edge $(s,r_j)$ for each
  $v_j = 1$ in a single batch of $O(d)$ updates.
  After that for each $u_i = 1$, we query the shortest path from $s$ to $l_i$.
  In total, we perform $O(m^{1 - \delta})$ queries and only use a single
  batch update consisting of $O(d)$ insertions.
  We claim that one of the queries returns less than $4$ iff $u^t Mv = 1$.

  First assume that $u^t M v = 1$. Then there exist indices $i, j$ such that
  $u_i = M_{ij} = v_j = 1$. Hence, there must be a path $l_i \to r_j \to s$
  of length $2$ in $G_M$ and since $l_i = 1$ we ask the query for $l_i$.
  Hence, any $(2-\varepsilon)$-approximation algorithm must return less than $4$ in the query for $l_i$.
  Now assume that a query for vertex $l_i$ returns less than $4$.
  Since any path from a vertex in $L$ to $s$ must have length $2k$ for some $k \in \mathbb{N}$,
  there exists a path $l_i \to r_j \to s$.
  Since query for $l_i$, we have $u_i = 1$. Due to the
  edge $l_i \to r_j$, $M_{ij} = 1$, and due to edge $r_j \to s$, $v_j = 1$.
  Thus, $u_i M_{ij} v_j = 1$ and $u^t M v = 1$.

  The proof for the decremental problem works by initially adding all edges from $s$ to $R$
  to the graph and then removing edges corresponding to the $0$-entries of $v$.
\end{proof}

To obtain the result of Theorem~\ref{Thm:SensitivityD} for $st$-shortest-paths,
observe that the above reduction can be changed to work for this problem:
We add additional vertices $s, t$ to the original bipartite graph
and connect $s$ and $t$ by a path of length $5$ (i.e., introducing $3$ additional vertices).
Then a similar proof to the above shows that any algorithm for $st$-reachability
that can distinguish between a shortest path of length at most $3$ or at least $5$
can be used to decide if $u^t M v = 1$.
The result for BW-Match follows from the reduction in \cite{abboud2014popular}.

\section{No (globally) fixed constant sensitivity with polynomial preprocessing time}
\label{Sec:NoFixedConstantSETH}

In this section, we will first discuss how SETH reductions depend on their sensitivity.
In particular, we will see that even though some of our previous reductions had
constant sensitivities, these constants are not bounded globally.
Afterwards we will prove that if we allow for polynomial preprocessing time,
then there cannot be a globally fixed upper bound on the sensitivities.

\subsection{A note on SETH reductions}
Let us recall that the sensitivities of the reductions in Section~\ref{Sec:NewSETHLowerBounds} 
had the form $K = c/\delta$, where $cn$ was the number of clauses
in the $k$-CNFSAT formula and $\delta < 1$ was a parameter
indicating the size of our initial graph. We will first discuss the dependency of $K$ on
$\delta$ and after that on $c$.

Firstly, let us discuss how the SETH reductions depend on the parameter $\delta$.
Let $F$ be a $k$-CNFSAT formula with $n$ variables and $O(n)$ clauses.
Notice that if the input graph that we construct in the reduction had size $N = O(2^{n/2})$,
then the preprocessing time of an algorithm refuting SETH would have to be $O(N^{2 - \varepsilon})$.
In order to allow for arbitrary polynomial preprocessing times of the algorithm,
the reductions in \cite{abboud2014popular} (and also the ones from the section before)
were parameterized such that the initial graphs have size
$O(2^{\delta n})$ (disregarding $\poly(n)$ factors).
Then for an algorithm with preprocessing time $O(N^t)$ we can pick
$\delta < 1/t$ and hence the preprocessing takes time
$O( 2^{\delta n t}) = O(2^{(1-\gamma)n})$ for some $\gamma > 0$.
Thus, despite the ``large'' preprocessing time we could still refute the SETH.
However, the sensitivities $K = c/\delta$ are not bounded if we consider $\delta$
as a parameter, i.e.\ $K \to \infty$ as $\delta \to 0$.
If we consider $\delta$ as the inverse of the power of the preprocessing time
$O(N^t)$, i.e.\ $\delta = 1/t$, then this can be interpreted as
``large preprocessing time'' corresponds to ``large sensitivity''
(since $\delta \to 0$ iff $t \to \infty$).

Now one might want to fix $\delta$ (and thus bound the preprocessing time)
in order to get (globally) fixed constant sensitivities.
Unfortunately, this approach is also not feasible; in the SETH reductions
for proving an algorithm that solves $k$-CNFSAT in time $2^{(1-\varepsilon)n}$,
there is another parameter $c$ which denotes the number of clauses of
the $k$-CNFSAT instance after the application of the sparsification lemma by
Impagliazzo, Paturi and Zane~\cite{impagliazzo2001which}. Looking at the proof of the lemma
one can see that $c = c(k,\varepsilon)$ and that
$c \to \infty$ as $k \to \infty$, as well as $c \to \infty$ as $\varepsilon \to 0$.
Hence, if we use the sparsification lemma and fix $\delta$, then we still cannot refute the SETH.

\subsection{Upper bound in case of polynomial preprocessing}
The former reasoning heavily relies on the constructions in our
proofs and does not rule the possibility of a reduction with (globally)
fixed constant sensitivity from SETH.
However, the following theorem and the corollary after it
will show that if we allow for polynomial preprocessing time,
then there cannot be a polynomial time upper bound on update
and query time under any conjecture.

\begin{theorem}
  \label{Thm:FixedEditDistanceAlgo}
  Let $P$ be a fully dynamic graph problem\footnote{We 
		only argue about dynamic graph problems to simplify
		our language. The theorem holds for all dynamic problems with the given properties.}
  with sensitivity $K$
  on a graph $G = (V,E)$ with $n$ vertices and $m$ edges
  where edges are added and removed. Assume that for the static version of $P$
  there exists an algorithm running in time $\poly(n)$.

  Then there exists an algorithm with $p(n) = n^t$ for some $t = t(K)$,
  $u(n) = O(K)$ and $q(n) = \log(n)$.
  The algorithm uses space $O(n^t)$.
\end{theorem}
\begin{proof}
  The basic idea of the algorithm is to preprocess the results of all possible queries that might
  be encountered during $P$. Since we have only sensitivity $K$, we will only have polynomially many
  graphs during the running time of the algorithm. We can save all of these results in a balanced tree
  of height $O(\log n)$ and during a query we just traverse the tree in logarithmic time.

  Denote the initial input graph by $G_0$.
  
  Since we know that after each update has size at most $K$ and after the update we roll back to the
  initial graph, we can count how many graphs can be created while running $P$.
  Particularly, notice that a graph $G_1 = (V,E_1)$ can be created while running $P$ if and only if
  $S = E_0 \triangle E_1$ has $|S| \leq K$.\footnote{$A \triangle B$ denotes the symmetric
  difference of $A$ and $B$, i.e.\ $A \triangle B = A \setminus B \cup B \setminus A$.}
  Then in total the number of graphs which $P$ will have to answer queries on can be bounded by computing
  how many such sets $S$ exist:
  \begin{align*}
    \sum_{i = 0}^K \binom{n^2}{i} \leq K \max_{i = 0, \dots, K} \binom{n^2}{i} = K \poly(n) = n^t,
  \end{align*}
  for some $t$ (where in the second step we used that $K$ is constant).
  
  Now consider the na\"ive algorithm which just preprocesses all trees and stores the results:
  We enumerate all $n^t$ possible trees and run the static algorithm on them.
  This can be done in time $O(\poly(n))$.
  We store the results of the static algorithm in a balanced binary tree with $O(n^t)$ leaves
  which is of height $O(\log n)$. The traversal of the tree can be done, e.g., in the following way:
  We fix some order $\prec$ on $V \times V$; this implicitly gives an order $\prec$ on the set $\{ S \subseteq V \times V \}$.
  For a graph $G_1$ we compute $S = E_1 \triangle E_0$ and traverse according to $S$.

  During updates the algorithm maintains an array of size $K$ which contains
  the edges that are to be removed or added ordered by $\prec$. This can be done
  in time $O(K) = O(1)$.

  During a query the algorithm will traverse the binary tree from the preprocessing according to
  $\prec$ and the updates that were saved during the updates. This takes time $O(\log n)$.

  Hence, we have found algorithm with $p(n) = \poly(n)$ and $u(n) = O(1)$ and $q(n) = O(\log n)$.
\end{proof}

\begin{corollary}
\label{Cor:FixedEditDistanceCorollary}
  If for a problem with the properties from Theorem~\ref{Thm:FixedEditDistanceAlgo}
  there exists a reduction from conjecture $\mathcal{C}$
  to $P$ with $p(n) = \poly(n)$ and $\max\{u(n),q(n)\} = \Omega(n^{\gamma - \varepsilon})$
  for any $\gamma > 0$ and all $\varepsilon \in (0,\gamma)$.
  Then $\mathcal{C}$ is false.
\end{corollary}
\begin{proof}
  We use Theorem~\ref{Thm:FixedEditDistanceAlgo} to obtain an algorithm
  which is better than the lower bound given in the assumption of the corollary.
  Hence, we obtain a contradiction to $\mathcal{C}$.
\end{proof}

Notice that Corollary~\ref{Cor:FixedEditDistanceCorollary} implies that in order to obtain
meaningful lower bounds for dynamic problems with a certain sensitivity, we must either
bound the preprocessing time of the algorithm or bound the space usage of the algorithm
or allow the sensitivity to become arbitrarily large.

\end{document}